\documentclass[11pt,a4paper]{amsart}

\usepackage{amsfonts,amssymb,amsmath,eucal,pinlabel,array,hhline}
\usepackage{slashed}
\usepackage[all]{xy}
\usepackage{tabulary}
\usepackage{fancyhdr}
\usepackage{a4wide}
\usepackage{calrsfs,bbm}

\newtheorem{theorem}{Theorem}[section]
\newtheorem*{thmintro}{Theorem}
\newtheorem*{corintro}{Corollary}

\newtheorem{lemma}[theorem]{Lemma}
\newtheorem{proposition}[theorem]{Proposition}
\newtheorem{corollary}[theorem]{Corollary}
\theoremstyle{definition}\newtheorem{definition}{Definition}
\newtheorem*{example}{Example}
\theoremstyle{remark}

\newenvironment{romanlist}
        {\begin{enumerate}
        }
        {\end{enumerate}}

\newcounter{ticklistc}

\newcommand{\Z}{\mathbb Z}
\newcommand{\R}{\mathbb R}
\newcommand{\C}{\mathbb C}

\newcommand{\E}{\mathcal E}
\newcommand{\EE}{\mathbb E}

\newcommand{\Di}{\slashed{D}}

\renewcommand{\L}{\mathcal L}

\newcommand{\Hom}{\mathrm{Hom}}

\newcommand{\G}{\Gamma}
\newcommand{\SI}{\Sigma}

\newcommand{\Arf}{\mathrm{Arf}}

\newcommand{\KW}{\mathit{KW}}
\renewcommand{\Pr}{\mathit{Pr}}
\renewcommand{\Im}{\mathrm{Im}}
\renewcommand{\Re}{\mathrm{Re}}
\newcommand{\rot}{\mathit{rot}}

\begin{document}

\title{Kac-Ward operators, Kasteleyn operators, and s-holomorphicity on arbitrary surface graphs}

\author{David Cimasoni}
\address{Section de math\'ematiques, 2-4 rue du Li\`evre, 1211 Gen\`eve 4, Switzerland}
\email{David.Cimasoni@unige.ch}
\thanks{This research was partially supported by the Swiss NSF}


\begin{abstract}
The conformal invariance and universality results of Chelkak-Smirnov on the two-dimensional Ising model hold for isoradial planar graphs with critical
weights. Motivated by the problem of extending these results to a wider class of graphs, we define a generalized notion of s-holomorphicity for functions
on arbitrary weighted surface graphs. We then give three criteria for s-holomorphicity involving the Kac-Ward, Kasteleyn, and
discrete Dirac operators, respectively. Also, we show that some crucial results known to hold in the planar isoradial case extend
to this general setting: in particular, spin-Ising fermionic observables are s-holomorphic, and it is possible to define a discrete version of the integral
of the square of an s-holomorphic function. Along the way, we obtain a duality result for Kac-Ward determinants on arbitrary weighted surface graphs.
\end{abstract}

\maketitle

%\tableofcontents

\pagestyle{myheadings}
\markboth{D. Cimasoni}{Kac-Ward operators, Kasteleyn operators, and s-holomorphicity on arbitrary surface graphs}

%%%%%%%%%%%%%%%%%%%%%%%%%%%%%%%%%%%

\section{Introduction}

\subsection{Motivation}

The Ising model is certainly one of the most famous models in statistical physics. It was introduced by Lenz in 1920~\cite{Lenz} as an attempt to understand Curie's temperature
for ferromagnets. It can be defined as follows.
Let $\Gamma$ be a finite graph with vertex set $V(\Gamma)$ and edge set $E(\Gamma)$. A {\em spin configuration} on $\Gamma$ is an element $\sigma$ of $\{-1,+1\}^{V(\Gamma)}$.
Given a positive edge weight system $J=(J_e)_{e\in E(\Gamma)}$ on $\Gamma$, the {\em energy} of such a spin configuration $\sigma$ is defined by
\[
\mathcal{H}(\sigma)=-\sum_{e=\{u,v\}\in E(\Gamma)}J_e\sigma_u\sigma_v.
\]
Fixing an {\em inverse temperature} $\beta\ge 0$ determines a probability measure on the set $\Omega(\Gamma)$ of spin configurations by
\[
\mu_{\Gamma,\beta}(\sigma)=\frac{e^{-\beta\mathcal{H}(\sigma)}}{Z_\beta^J(\Gamma)},
\]
where the normalization constant
\[
Z_\beta^J(\Gamma)=\sum_{\sigma\in\Omega(\Gamma)}e^{-\beta\mathcal{H}(\sigma)}
\]
is called the {\em partition function\/} of the {\em Ising model on $\Gamma$ with coupling constants $J$\/}. The name of the model comes from a student of Lenz, Ernst Ising, who proved in his PhD Thesis the absence of phase transition in dimension one, i.e. in the case of $\Gamma=\mathbb{Z}$ and $J_e$ independent of $e$. He also conjectured the same behavior in higher dimensions, a conjecture later disproved by Peierls~\cite{Pei}. Using their celebrated duality argument, Kramers and Wannier~\cite{K-W} then gave a heuristic
derivation of the critical temperature in the case of the square lattice $\Gamma=\mathbb{Z}^2$. 

The rigorous proof of this later result by Onsager in 1944 together with his exact computation of the partition function~\cite{Ons} led to an explosion of activity in the field.
In particular, Kac and Ward~\cite{KW} tried to find a more combinatorial approach to the results of Onsager. They defined a matrix $\KW(\G)$ with rows and columns indexed by the set
$\EE$ of oriented edges of the square lattice $\G$, whose determinant coincides with the square of $Z^\mathit{Ising}(\G)$,
the high temperature expansion of the Ising partition function on $\G$. However, several arguments in~\cite{KW} are of heuristic nature, and some key
topological statement turned out not to hold~\cite{She}. Since the technicalities raised by a rigorous proof of this equality seemed formidable, the focus shifted to finding
combinatorial methods not involving directly this {\em Kac-Ward matrix\/}. This was achieved independently by Hurst-Green~\cite{H-G}, Kasteleyn~\cite{Ka2} and Fisher~\cite{Fi2}:
they related $Z^\mathit{Ising}(\G)$ with the dimer partition function $Z^\mathit{dimer}(F_\G)$ on an associated {\em Fisher graph\/} $F_\G$, and found a skew-symmetric adjacency matrix $K(F_\G)$
-- the associated {\em Kasteleyn matrix\/} -- whose Pfaffian gives $Z^\mathit{dimer}(F_\G)$. This method was later extended by Kasteleyn~\cite{Ka3} to any planar graph,
and by various authors to the general case of surface graphs~\cite{Tes,G-L,C-RI}. Note that the first direct combinatorial proof of the Kac-Ward formula for any planar graph was only obtained in
1999 by Dolbilin {\em et al.\/}~\cite{DZMSS}, and recently extended by the author to surface graphs~\cite{Cim2}.

There are several issues with this {\em Fisher correspondence\/} $\G\mapsto F_\G$, the main one being that it does not preserve crucial geometric and combinatorial properties
of the graph. For example, if $\G$ is an {\em isoradial\/} graph, i.e. if each face is inscribed into a circle of fixed radius, then $F_\G$ will not be isoradial in general. Also, even if $\G$ is {\em bipartite\/},
i.e. if its vertices split up into two sets $B(\G)\cup W(\G)$ such that the edges never link vertices in the same set, $F_\G$ will not be bipartite.
This is a problem, since many remarkable statements about the dimer model are known to hold only for bipartite~\cite{KOS} or isoradial graphs~\cite{Ken}... and yet, they did
seem to transfer to the Ising model~\cite{BdT1,BdT2}. The explanation to this mystery came in the recent paper of Dub\'edat~\cite{Dub}: he (re)discovered another mapping $\G\mapsto C_\G$~\cite{W-L},
where this time, $C_\G$ is always bipartite, and always isoradial if $\G$ is. This new mapping therefore permits the transfer of the whole power of the dimer technology to the Ising model.

As apparent from its definition, the Ising model can be studied on an arbitrary {\em abstract\/} weighted graph $(\G,J)$. However, as explained above, the effective computation of
its partition function requires the choice of an embedding of $\G$ in a surface. Furthermore, at criticality, the coupling constants $J$ are expected to correspond to a natural embedding of $\G$ such that the model exhibits properties of conformal invariance and universality at the scaling limit. This can be stated in a precise way in (at least) two different settings: the isoradial and biperiodic cases.
If the graph is isoradially embedded in the plane, then there are natural coupling constants such that the model is
critical at inverse temperature $\beta=1$. (The first rigorous proof of this fact was obtained recently by Lis~\cite{Lis2} using Kac-Ward matrices.) Note that graphs can be isoradially embedded in flat surfaces of arbitrary genus (with cone-type singularities), and that such a flat metric defines a conformal structure on the surface
(see~\cite{Tro}). Therefore, it does make sense to talk about critical embeddings of graphs in surfaces of arbitrary topology (a situation first considered by Mercat~\cite{Mer}),
and to conjecture conformal invariance and universality of the model at the scaling limit. For example, it is believed that at the critical temperature and in the scaling limit,
the Ising partition function satisfies
\[
\log Z_{\beta_c}^J(\Gamma)\simeq f|V(\G)| + h\log|V(\G)| +\mathsf{fsc},
\]
where $f$ is the free energy, $h$ depends on the topology and on the singularities of the metric on $\Sigma$, and $\mathsf{fsc}$ is a universal term which only depends on the conformal
structure on $\Sigma$. (We refer to~\cite{CSM2} for numerical evidence in the case of triangular lattices embedded in a genus 2 surface.)
If the graph is planar and biperiodic, i.e. invariant under the action of a lattice $L\simeq\mathbb{Z}^2$, then the critical temperature can be determined~\cite{Li2,C-D}.
(Note that the second reference uses Kac-Ward matrices once again, as well as the correspondence $\G\mapsto C_\G$ mentioned above.) Also, the coupling constants determine a natural
conformal parameter for the corresponding torus $\C/L$, and once again, it is natural to conjecture conformal invariance and universality at the scaling limit.
(See Corollary~\ref{cor:tau} below where we use recent results on the dimer model~\cite{KSW} to check that $\log Z_\beta^J(\Gamma)$ behaves in the expected way in this setting.)

In the planar isoradial case, a milestone was reached by Smirnov and coauthors in a recent series of papers (see the review~\cite{DS11} and references therein).
In particular, Chelkak and Smirnov~\cite{CS09} introduced {\em fermionic observables\/} for the Ising model on any planar critical isoradial graph, and showed that on bounded domains
with appropriate boundary conditions, they converge to universal and conformally invariant limits.
But why do these authors restrict themselves to planar isoradial graphs, while the universality and conformal invariance are expected to hold at criticality on more general graphs?
The problem is that one of the crucial ingredients is a discrete theory of holomorphic functions, and isoradial graphs form
the widest class of graphs on which such a theory works sufficiently well~\cite{CS11}. For example, a reasonable discrete version of the $\bar\partial$-operator,
whose kernel consists of so-called {\em discrete holomorphic functions\/}, does not seem to exist on non-isoradial graphs. Moreover, this notion of
discrete holomorphicity does not suffice for the Ising model: Chelkak and Smirnov therefore introduced a stronger notion, known as {\em spin\/},
{\em strong\/}, or simply {\em s-holomorphicity\/}. Among the numerous technical results obtained in~\cite{CS09}, we quote the following:
\begin{romanlist}
\item{the fact that the fermionic observables are s-holomorphic;}
\item{the possibility to define a discrete version of $h(z)=\Im\int f(z)^2\mathit{dz}$ for s-holomorphic functions;}
\item{the subharmonicity and superharmonicity of $h$ on the original graph and on its dual.}
\end{romanlist}

In the present paper, we show that, unlike the original notion of discrete holomorphicity, the one of s-holomorphicity does extend to the most general case, that is, to arbitrary weighted graphs embedded in orientable surfaces. Obviously, the whole theory does not extend, but a surprisingly big amount does.
Let us also point out that our results do not simply consist in known facts extended from the isoradial to the general case: several statements where previously
unknown even in the isoradial case. Therefore, it is our hope that this article will be of interest even to the reader that merely wishes to understand
why, in the isoradial case, s-holomorphic functions and fermionic observables are natural objects.
Furthermore, our results are particularly appealing not only in the planar isoradial case, but also in the critical biperiodic case and for graphs isoradially embedded
in arbitrary surfaces. Therefore, we hope to be able to use them to eventually prove conformal invariance and universality results in these more general settings.

\subsection{Statement of the results}

Our results hold for an arbitrary graph $\G$ with edge weights $x\in[0,1]^E$ embedded in an orientable surface $\SI$.
At this level of generality, a geometric tool known as a {\em spin structure\/} is needed, which basically consists in a vector field $\lambda$ on $\SI$ with zeroes of even index in $\SI\setminus\G$.
This allows us to associate to each oriented edge $e\in\EE$ an ``argument'' $a_e$ measured with respect to $\lambda$.
The cases to keep in mind are the planar and toric ones, where $\lambda$ can be chosen constant;
the argument $a_e$ then simply gives the direction of $e$, in the sense that the oriented edge $e$ points in the direction of $\exp(ia_e)$ with respect to $\lambda$.

Let us parametrize the edge weight $x_e\in [0,1]$ by $x_e=\tan(\theta_e/2)$. Then, we say that a function $F$ defined on the set
$\diamondsuit=\{z_e\}_{e\in E}$ of middle points of the edges of $\G$ is {\em s-holomorphic\/} around the vertex $v$ of $\G$ if, for any oriented edge
$e\in\EE$ originating at $v$,
\[
\Pr\left(F(z_e);\left[i\exp(i(a_e+\theta_e))\right]^{-\frac{1}{2}}\right)=
\Pr\left(F(z_{e'});\left[i\exp(i(a_{e'}-\theta_{e'}))\right]^{-\frac{1}{2}}\right)\exp\left(\textstyle{\frac{i}{2}(\beta_e-\theta_e-\theta_{e'}})\right),
\]
where $\Pr(-;u)$ denotes the orthogonal projection onto $u\cdot\R$, $e'\in\EE$ is the edge following $e$ in counterclockwise order around $v$, and
$\beta_e$ is the oriented angle from $e$ to $e'$. (See Definition~\ref{def:s} for details.) Note that in the critical isoradial case, $\beta_e-\theta_e-\theta_{e'}$ vanishes and the
equality above defines the usual notion of s-holomorphicity.

Here is one of our main results. (See Theorem~\ref{thm:s} and Corollary~\ref{cor:s} for the full statement.)

\begin{thmintro}
There are explicit $\R$-linear injective maps $S\colon\C^\diamondsuit\to\C^\EE$ and $T\colon\C^\diamondsuit\to\C^{B(C_\G)}$ such that,
given $F\in\C^\diamondsuit$, the following are equivalent:
\begin{romanlist}
\item{$\exp\left(i\textstyle{\frac{\pi}{4}}\right)F$ is s-holomorphic.}
\item{$S(F)$ lies in the kernel of the Kac-Ward operator on $\G$.}
\item{$T(F)$ lies in the kernel of the Kasteleyn operator on $C_\G$.}
\end{romanlist}
Furthermore, if $\G$ is isoradially embedded with critical weights, then there is an explicit $\R$-linear injective map $T'\colon\C^\diamondsuit\to\C^{B(C_\G)}$
such that these conditions are equivalent to:
\begin{romanlist}{\setcounter{enumi}{3}}
\item{$T'(F)$ lies in the kernel of the discrete $\overline{\partial}$-operator on $C_\G$.}
\end{romanlist}
\end{thmintro}

Note that in the planar isoradial case, the equivalence of $(i)$ and $(ii)$ is due to Lis~\cite{Lis}.
Also, the equivalence of $(i)$ and $(iii)$ can be understood as a generalization of the ``propagation equation'' of~\cite[Section 3.2]{CS09}
(see also~\cite[Section 4.3]{Mer} and~\cite[Section 4.2]{Dub}).

Our proof of the equivalence of $(ii)$, $(iii)$ and $(iv)$ relies on explicit relations between the corresponding operators (Section~\ref{sec:rel}).
We believe that some of these results are of independent interest. In particular, Theorem~\ref{thm:corr}, which relates the Kac-Ward operator on $\G$
with the Kasteleyn operator on $C_\G$, also immediately implies the following generalized Kramers-Wannier duality.
(The original Kramers-Wannier duality corresponds to the planar case.)

\begin{corintro}
For any weighted graph $(\G,x)\subset\SI$ and any character $\varphi\colon\pi_1(\SI)\to\C^*$,
\[
2^{|V(\G)|}\prod_{e\in E(\G)}(1+x_e)^{-1} \det(\KW^\varphi(\G,x))=2^{|V(\G^*)|}\prod_{e\in E(\G)}(1+x^*_e)^{-1} \det(\KW^\varphi(\G^*,x^*)),
\]
where the weights $x$ and $x^*$ are related by $x+x^*+xx^*=1$.
\end{corintro}

In the genus one case, this same Theorem~\ref{thm:corr} allows us to transfer the recent results of Kenyon-Sun-Wilson~\cite{KSW} from the dimer to the Ising model.
(See Corollary~\ref{cor:tau} below for the full statement.)

\begin{corintro}
Let $(\G,J)$ be a weighted finite toric graph, interpreted as the quotient of a biperiodic planar graph $\mathcal{G}$ by a lattice $L\simeq\Z^2$, and set
$\G_n=\mathcal{G}/n L$. At the critical inverse temperature $\beta=\beta_c$, the Ising partition function on $\G_n$ satisfies
\[
\log Z^J_{\beta_c}(\G_n)= n^2 f(\beta_c) + \mathsf{fsc(\tau)} +o(1),
\]
where $f$ is the free energy per fundamental domain, $\tau\in\mathbb{H}$ can be explicitely computed from $\det(\KW^\varphi(\G_1,x))$, and $\mathsf{fsc}(\tau)$
is a universal function of the modular parameter $\tau$ which is invariant under modular transformations.
\end{corintro}

Coming back to s-holomorphicity, the theorem stated above should convince the reader that this notion is ``natural'', as it is linked to the operators
appearing in the theory. But is there any non-trivial example of s-holomorphic functions? More precisely, does there exist some generalized fermionic
observables that are s-holomorphic? This is the case, and these observables are basically given by the cofactors of the Kac-Ward
matrix (a fact obtained independently by Lis~\cite{Lis} in the planar case). For simplicity, we shall only state a corollary and
not be very precise with the definitions, referring to Theorem~\ref{thm:inv} for the full statement of the main result, and to subsection~\ref{sub:inv}
for more details.

\begin{corintro}
Given any fixed oriented edge $e_0\in\EE$, consider the function $F_{e_0}\in\C^{\diamondsuit}$ defined by
\[
F_{e_0}(z)=\frac{\exp\left(i\textstyle{\frac{\pi}{4}}\right)}{\cos(\theta_e/2)}\sum_{\xi\in\E(e_0,z)}(-1)^{q_\lambda(\gamma^0_\xi)}\exp\left(-\textstyle{\frac{i}{2}\rot_\lambda(\gamma_\xi)}\right)\textstyle{\prod_{e'\in\xi}}x_{e'},
\]
where $z=z_e\neq z_{e_0}$, $\E(e_0,z)$ denotes the set of subgraphs $\xi$ consisting in a collection $\gamma^0_\xi$ of loops together with one path $\gamma_\xi$
from $e_0$ to $z$, $\rot_\lambda(\gamma_\xi)$ is the total rotation angle of the path $\gamma_\xi$ (measured with respect to the fixed vector field $\lambda$),
and $q_\lambda$ is the quadratic form associated with $\lambda$. Then, $F_{e_0}$ is s-holomorphic around every vertex of $\G$ not adjacent to $e_0$.
\end{corintro}

This result also has a more direct link to the Ising model. Indeed, consider a biperiodic planar graph $\mathcal{G}$
with edge weights $J=(J_e)_e\in(0,\infty)^{E(\mathcal{G})}$. From Theorem~\ref{thm:inv} and~\cite{C-D}, we immediately get:

\begin{corintro}
There exists a non-trivial biperiodic s-holomorphic function on the weighted graph $(\mathcal{G},x)$, where $x_e=\tanh(\beta J_e)$, if and only if $\beta$ is
the critical inverse temperature for the Ising model on $(\mathcal{G},J)$.
\end{corintro}

So, many natural functions are s-holomorphic (except possibly at a couple of vertices), and the existence of functions that are s-holomorphic everywhere
is related, at least in the case of biperiodic graphs, to criticality.

Finally, we show that the possibility to define a discrete version of $h(z)=\Im\int f(z)^2\mathit{dz}$ for s-holomorphic functions does extend to our setting (Proposition~\ref{prop:F2}). However, in general, this $h$ is not subharmonic on $\G$ as it is in the critical isoradial case. So this is the point where the extension of the proof of Chelkak-Smirnov
breaks down, and where further original ideas are required. We do believe however that the theory of generalized s-holomorphic functions initiated in the present
article is a first step in this direction. 

\subsection{Organisation of the article}

In Section~\ref{sec:def}, we present the main objects involved in the article and settle many notations: weighted surface
graphs, isoradial graphs, discrete Laplacians, Kac-Ward operators, Kasteleyn operators, and Dirac operators are introduced in this order in separate paragraphs.

In Section~\ref{sec:rel}, we exhibit natural relations between these operators. The most technical such result is Theorem~\ref{thm:corr}, which relates Kac-Ward
and Kasteleyn operators, implies duality statements (Corollaries~\ref{cor:KW1} and~\ref{cor:KW2}), and allows us to transfer the full dimer technology to the Ising model
(Corollary~\ref{cor:tau}). The remaining subsections, which only make sense
in the isoradial case, relate the Kasteleyn and discrete Dirac operators (Proposition~\ref{prop:Kast-D}), the discrete Dirac and Laplace operators (Proposition~\ref{prop:D-L2}), and the discrete Dirac operators on $C_\G$ and on the double of $\G$ (Proposition~\ref{prop:d-d}).

Finally, Section~\ref{sec:s-holo}
contains our main results. In subsection~\ref{sub:ker}, we analyse the kernel of the Kac-Ward operator (Proposition~\ref{prop:ker}). In subsection~\ref{sub:3}, we give the definition of s-holomorphicity, and building on the results of Section~\ref{sec:rel}, we prove our three criteria for s-holomorphicity (Theorem~\ref{thm:s}). In subsection~\ref{sub:inv}, we analyse the cofactors of the Kac-Ward matrix (Theorem~\ref{thm:inv}), thus obtaining the s-holomorphicity of generalized fermionic observables (Corollary~\ref{cor:s-holo}). Finally, in subsection~\ref{sub:F2}, we show that a discrete version of
the integral of the square of an s-holomorphic function can be defined (Proposition~\ref{prop:F2}).

%%%%%%%%%%%%%%%%%%%%%%%%%%%%%%%%%%%

\section{The graphs and operators involved}
\label{sec:def}

The aim of this section is to present the main objects involved in this article, and for some of them, to recall their role in statistical mechanics.
This section is therefore purely expository, and does not contain any new result. It is organised as follows. We first introduce the general setup of weighted surface graphs, presenting in particular
the construction of the bipartite graph $C_\G$ associated to any surface graph $\G$, which will play a crucial role in this article (subsection~\ref{sub:graph}). Then, we recall the setup of isoradial
graphs embedded in flat surfaces, a particularly interesting and well-studied class of surface graphs. Finally, in subsections~\ref{sub:Lapl} to~\ref{sub:Dirac}, we define the operators that will be
studied in the rest of the article: the Laplace, Kac-Ward, and Kasteleyn operators -- defined on any weighted surface graph -- and the discrete Dirac operator, defined only on isoradial graphs.

\subsection{Weighted surface graphs}
\label{sub:graph}

Let us start by setting up our notations for graphs. The set of vertices (resp. edges) of a graph $\G$ will be denoted by $V(\G)$ (resp. $E(\G)$). We shall write
$\EE(\G)$ for the set of oriented edges of $\G$. (Each element of $E(\G)$ thus corresponds to two elements of $\EE(\G)$.)
Following~\cite{Ser}, we shall denote by $o(e)$ the origin of an oriented edge $e\in\EE(\G)$, by $t(e)$ its terminus, and by $\bar{e}$ the same edge with the opposite orientation.
By abuse of notation, we shall also write $e\in E(\G)$ for the unoriented edge of $\G$ corresponding to the oriented edges $e,\bar{e}\in\EE(\G)$.
Finally, we shall write $\EE(\G)_v$ for the set of oriented edges of $\G$ with origin a fixed vertex $v$. 

By a {\em weighted surface graph\/}, we mean a finite graph $\Gamma$ endowed with edge weights $x=(x_e)_{e\in E(\G)}\in[0,1]^{E(\G)}$, embedded in a compact connected orientable surface $\Sigma$
so that the complement of $\G$ in $\SI$ is the disjoint union of topological discs, which we call {\em faces\/}.

The {\em dual\/} of a weighted surface graph $(\G,x)\subset\SI$ is the weighted surface graph $(\G^*,x^*)\subset\SI$ obtained as follows: each face of $\G\subset\SI$ defines a vertex of $\G^*$,
and each edge of $\G$ bounding two faces of $\G\subset\SI$ defines an edge between the two corresponding vertices of $\G^*$. More precisely, an oriented edge $e\in\EE(\G)$ defines $e^*\in\EE(\G^*)$,
the dual oriented edge obtained by turning the oriented edge $e$ counterclockwise. (Since $\SI$ is orientable, this can be done in a consistent way.) Note that $(\bar{e})^*=\overline{e^*}$.
As for the {\em dual weights} $x^*\in[0,1]^{E(\G^*)}$, they are defined via the equality $x+x^*+xx^*=1$. If we use the parametrization $x=\tan(\theta/2)$ with $\theta\in[0,\pi/2]$,
then $\theta$ and $\theta^*$ are simply related by $\theta+\theta^*=\pi/2$. Note that $((\G^*)^*,(x^*)^*)=(\G,x)$.

\begin{figure}[Htb]
\labellist\small\hair 2.5pt
\pinlabel {$\G$} at 280 370
\pinlabel {$D_\G$} at 1100 370
\pinlabel {$z_e$} at 1190 150
\pinlabel {$\tan(\theta_e/2)$} at 300 215
\pinlabel {$\sin(\theta_e)$} at 1110 210
\pinlabel {$\sin(\theta_e)$} at 1335 210
\pinlabel {$\cos(\theta_e)$} at 1290 260
\pinlabel {$\cos(\theta_e)$} at 1290 110
\endlabellist
\begin{center}
\includegraphics[width=\textwidth]{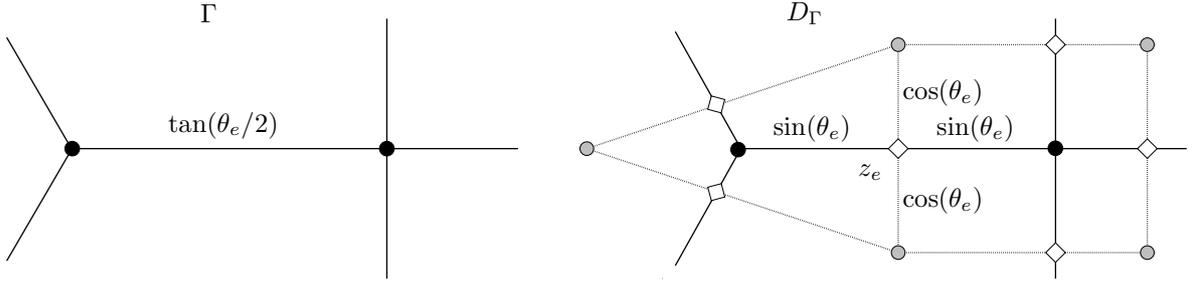}
\caption{The weighted bipartite graph $D_\G$ associated to the weighted graph $\G$.}
\label{fig:double}
\end{center}
\end{figure}

The {\em double\/} of $(\G,x)\subset\SI$ is the weighted bipartite surface graph $(D,y)=(D_\G,y(x))\subset\SI$ given as follows: as a subset of $\SI$, $D_\G=\G\cup\G^*$,
with vertex set $B(D)\cup W(D)$, where $B(D)=V(\G)\cup V(\G^*)=:\Lambda$ and $W(D)=E(\G)\cap E(\G^*)=:\diamondsuit$. As for the weighted edges, an edge $e$ of $\G$ with weight $x_e=\tan(\theta_e/2)$
will give rise to two edges of $D_\G$ with weight $\sin(\theta_e)$ and to two transverse edges of $D_\G$ with weight $\cos(\theta_e)$. Following~\cite{CS09}, we shall denote by $z_e\in\diamondsuit$
the vertex of $D_\G$ corresponding to the edge $e\in E(\G)$. This is illustrated in Figure~\ref{fig:double}.

Following Wu-Lin~\cite{W-L} (see also~\cite{BdT,Dub}), we shall consider another weighted surface graph $(C,y)=(C_\G,y(x))\subset\SI$ associated to a given weighted surface graph $(\G,x)\subset\SI$.
It is obtained as follows: replace each edge $e$ of $\G$ by a rectangle with the edges parallel to $e$ having weight $\sin(\theta_e)$ while the other two edges have weight $\cos(\theta_e)$.
In each corner of each face of $\G\subset\SI$, we now have two vertices; join them with an edge of weight $1$. This is illustrated in Figure~\ref{fig:C_G}.
Note that since the surface $\SI$ is orientable, the graph $C_\G$ is bipartite. Note also that the weighted graph $(C_{\G^*},y(x^*))$ associated to $(\G^*,x^*)$ is equal to the weighted graph
$(C_\G,y(x))$ associated to $(\G,x)$.

\begin{figure}[Htb]
\labellist\small\hair 2.5pt
\pinlabel {$\G$} at 280 370
\pinlabel {$C_\G$} at 1200 370
\pinlabel {$\tan(\theta_e/2)$} at 300 215
\pinlabel {$\sin(\theta_e)$} at 1225 250
\pinlabel {$\sin(\theta_e)$} at 1225 115
\pinlabel {$\cos(\theta_e)$} at 1080 0
\pinlabel {$\cos(\theta_e)$} at 1300 0
\pinlabel {$1$} at 1055 270
\pinlabel {$1$} at 1055 103
\pinlabel {$1$} at 1375 260
\pinlabel {$1$} at 1375 109
\endlabellist
\begin{center}
\includegraphics[width=\textwidth]{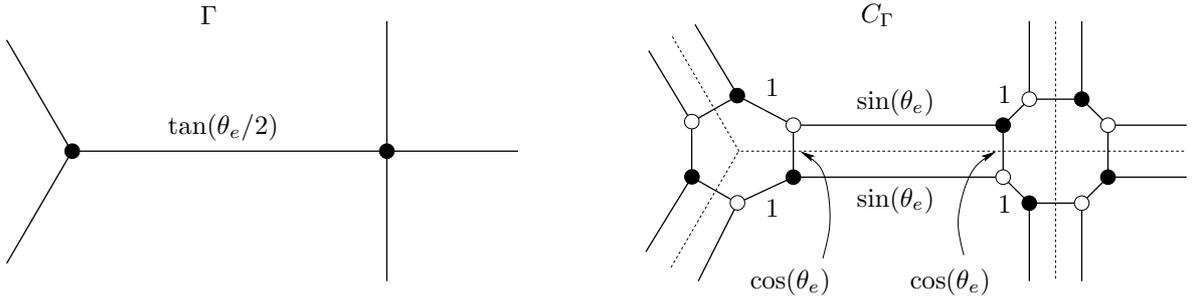}
\caption{The weighted bipartite graph $C_\G$ associated to the weighted graph $\G$.}
\label{fig:C_G}
\end{center}
\end{figure}

In our considerations, it will be useful to endow the surface $\SI$ with a character of its fundamental group, that is, an element of
\[
\Hom(\pi_1(\SI),\C^*)=H^1(\SI;\C^*).
\]
If $\G\subset\SI$ is a surface graph, then such a cohomology class can be represented by a 1-cocycle on $\G$. Let us recall that, with the notations above,
a {\em 1-cochain\/} is a map $\varphi\colon\EE(\G)\to\C^*$ such that $\varphi(\bar{e})=\varphi(e)^{-1}$ for all $e\in\EE(\G)$.
It is called a {\em 1-cocycle\/} if for each face $f$ of $\G\subset\SI$, $\varphi(\partial f):=\prod_{e\in\partial f}\varphi(e)=1$.
Multiplying each $\varphi(e)$ such that $o(e)=v$ by a fixed element of $\C^*$ results in another 1-cocycle, which is said to be {\em cohomologous\/} to $\varphi$.
Equivalence classes of 1-cocycles define the first cohomology group $H^1(\SI;\C^*)$, which only depends on $\SI$ (and not on $\G$): if $\SI$ is a closed surface of genus $g$, then
$H^1(\SI;\C^*)\simeq(\C^*)^{2g}$.

If $\Gamma$ is endowed with a 1-cochain $\varphi$, define the associated 1-cochain $\varphi_C\colon\EE(C_\G)\to\C^*$ by
\[
\varphi_C(w,b)=
\begin{cases}
\varphi(e)&\text{if $(w,b)$ runs parallel to $e\in\EE(\G)$;}\\ 
1&\text{else,}
\end{cases}
\]
as illustrated in Figure~\ref{fig:cocycles}. Note that if $\varphi$ is a cocycle, then so is $\varphi_C$, and they induce the same cohomology class in $H^1(\SI;\C^*)$.
Also, any 1-cochain $\varphi_D$ on $D_\G$ naturally induces 1-cochains $\varphi$ on $\G$ and $\varphi_*$ on $\G^*$, as illustrated in Figure~\ref{fig:cocycles}.
Here again, if $\varphi_D$ is a cocycle, then so are $\varphi$ and $\varphi_*$, and all three induce the same cohomology class.

\begin{figure}[Htb]
\labellist\small\hair 2.5pt
\pinlabel {$e$} at 340 205
\pinlabel {$\varphi(e)=\varphi(\bar{e})^{-1}$} at 225 100
\pinlabel {$\varphi_2\varphi^{-1}_4$} at 1740 245
\pinlabel {$\varphi_1\varphi^{-1}_3$} at 1950 110
\pinlabel {$\varphi_1$} at 1435 190
\pinlabel {$\varphi_2$} at 1380 245
\pinlabel {$\varphi_3$} at 1210 190
\pinlabel {$\varphi_4$} at 1380 85
\pinlabel {$\varphi(e)$} at 720 285
\pinlabel {$\varphi(e)^{-1}$} at 730 35
\pinlabel {\footnotesize{$1$}} at 565 125
\pinlabel {\footnotesize{$1$}} at 905 190
\pinlabel {\footnotesize{$1$}} at 515 300
\pinlabel {\footnotesize{$1$}} at 505 10
\pinlabel {\footnotesize{$1$}} at 955 15
\pinlabel {\footnotesize{$1$}} at 965 310
\endlabellist
\begin{center}
\includegraphics[width=\textwidth]{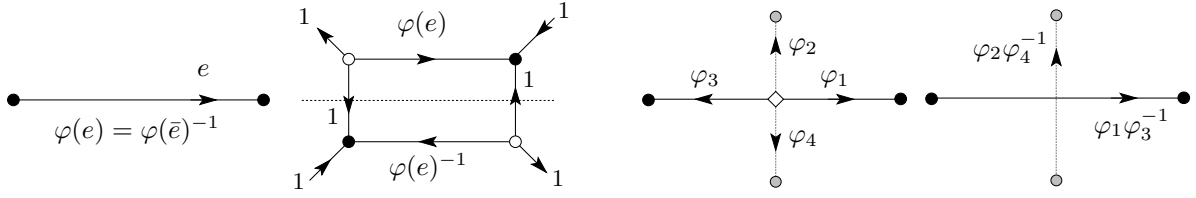}
\caption{A given cochain $\varphi$ on $\G$, and the associated cochain $\varphi_C$ on $C_\G$;
a cochain $\varphi_D$ on $D_\G$, and the associated cochains $\varphi$ on $\G$ and $\varphi_*$ on $\G^*$.}
\label{fig:cocycles}
\end{center}
\end{figure}

\subsection{Isoradial graphs}
\label{sub:iso}

The majority of the results of the present article hold for arbitrary weighted surface graphs, as defined above. However, many of these results take a particularly pleasant form when the graph
is isoradially embedded with critical weights. Let us now recall these concepts.

Consider a collection of planar rhombi of equal side length, say $\delta$, each rhombus having a fixed diagonal $e$ and corresponding half-rhombus angle $\theta_e\in[0,\pi/2]$
as illustrated below.

\begin{figure}[htb]
\labellist\small\hair 2.5pt
\pinlabel {$e$} at 300 205
\pinlabel {$\theta_e$} at 100 140
\endlabellist
\centerline{\psfig{file=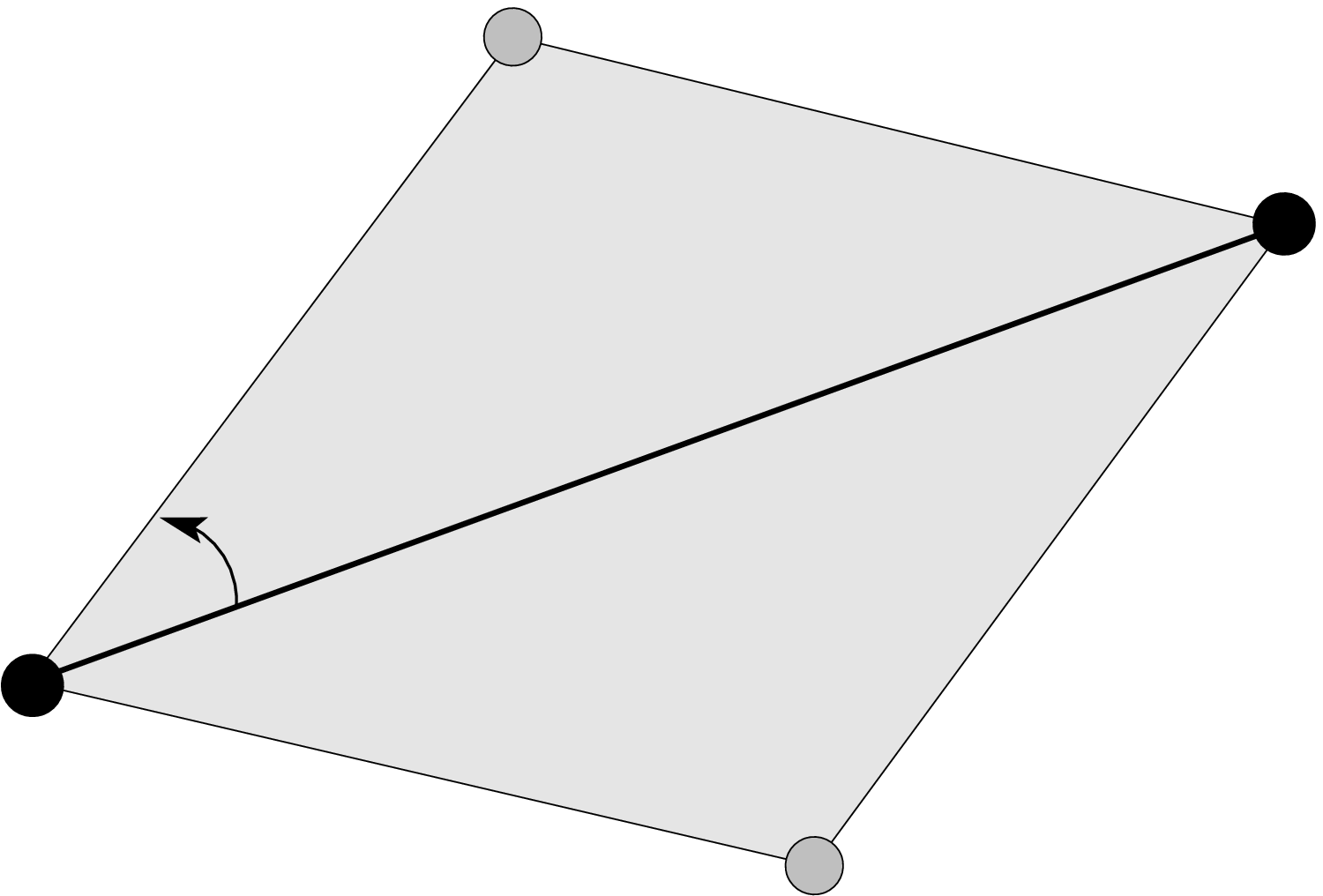,height=2cm}}
\end{figure}

\noindent Paste these rhombi together along their sides so that extremities of diagonals are glued to extremities of diagonals. These diagonals then form the edges of a graph $\G$ embedded
in a flat surface $\SI$, with so-called {\em cone-type singularities\/} in the vertex set $V(\G)$ of $\G$ and in the faces of $\G\subset\SI$. The cone angle of the singularity at $v\in V(\G)$ is given by
the sum of the angles of the rhombi adjacent to $v$ (and similarly for singularities in the faces).

\begin{definition}
\label{def:iso}
A graph $\G$ is {\em $\delta$-isoradially embedded\/} in a flat surface $\SI$ if it can be realized by pasting rhombi of side length $\delta$ as above, so that all the cone angles are
odd multiples of $2\pi$. The corresponding {\em critical weights\/} $x=(x_e)_{e\in E(\G)}$ are given by $x_e=\tan(\theta_e/2)$, where $\theta_e\in[0,\pi/2]$ denotes the half-rhombus angle
associated to the edge $e\in E(\G)$.
\end{definition}

Note that if $\SI$ is the (flat) plane, then this definition coincides with the usual notion of isoradial embedding as considered in~\cite{Ken,CS09}. It is however more general, as it allows us to
work in surfaces of arbitrary genus.

Given a flat surface $\SI$ with cone-type singularities, the parallel transport along a closed loop $\gamma$ in $\SI$ defines an element of $\mathit{SO(2)}=S^1$ called the {\em holonomy\/} of $\gamma$.
If all cone-angles are multiples of $2\pi$, then the holonomy can be described by a homomorphism $\mathit{Hol}\in\mathit{Hom}(\pi_1(\SI),S^1)=H^1(\SI;S^1)$. If this homomorphism is trivial,
$\SI$ is said to have {\em trivial holonomy\/}. (An example of genus $g\ge 1$ is given by the regular $4g$-gone with opposite sides identified.)
If all cone-angles are odd multiples of $2\pi$, then the inverse square roots of the holonomy can be described by $S^1$-valued 1-cocycles on $\SI$, which should be thought of
as {\em discrete spin structures\/} on $\SI$ (see~\cite[Section 3.1]{Cim1}).

Note that if $\G$ is $\delta$-isoradially embedded in a flat surface $\SI$, then so is the dual graph $\G^*$, while the associated bipartite graphs
$D_\G$ and $C_\G$ can be naturally $\frac{1}{2}\delta$-isoradially embedded in this same flat surface $\SI$ (some rhombi of $C_\G$ being degenerate).
Furthermore, if $\G$ is endowed with the critical weights $x$, then the associated weights $x^*$ on $\G^*$ are nothing but the critical weights determined by the corresponding isoradial embedding.
Finally, the associated weights $y$ on $D_\G$ and $C_\G$ are simply given by $y_e=\sin(\theta_e)$, with $\theta_e$ the half-rhombus angle determined by the isoradial embedding.
This is illustrated in Figure~\ref{fig:iso}.

\begin{figure}[Htb]
\labellist\small\hair 2.5pt
\pinlabel {$\theta$} at 123 191
\pinlabel {$\theta^*$} at 700 215
\pinlabel {$\theta$} at 1125 180
\pinlabel {$\theta$} at 1340 180
\pinlabel {$\theta$} at 1740 255
\pinlabel {$\theta$} at 1740 110
\pinlabel {$\theta^*$} at 1223 244
\pinlabel {$\theta^*$} at 1223 95
\pinlabel {$\theta^*$} at 1623 172
\pinlabel {$\theta^*$} at 1843 172
\endlabellist
\begin{center}
\includegraphics[width=\textwidth]{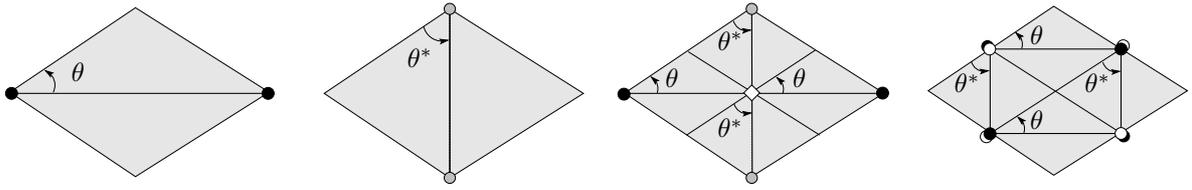}
\caption{If $\G$ is isoradially embedded, then so are $\G^*$, $D_\G$ and $C_\G$.}
\label{fig:iso}
\end{center}
\end{figure}

\subsection{The discrete Laplace operator}
\label{sub:Lapl}

Fix a surface graph $\Gamma\subset\SI$ endowed with a 1-cocycle $\varphi\colon\EE(\G)\to\C^*$. The space $\C^{V(\G)}$ can be thought of as a combinatorial version of the space of complex
valued functions on $\SI$, and more generally, of the space of sections of the complex line bundle induced by the class of $\varphi$ in $H^1(\SI;\C^*)$.
As first observed by Eckmann~\cite{Eck} (in the untwisted case),
the Laplace operator on this space can be naturally discretized by an operator $\Delta^\varphi_\G$ on $\C^{V(\G)}$, provided one fixes inner products on the cochain spaces
$C^0(\G;\C)$ and $C^1(\G;\C)$. Furthermore, if these inner products are diagonal (i.e. given by positive vertex weights $\{\mu_v\}_{v\in V(\G)}$ and edge weights $\{\nu_e\}_{e\in E(\G)}$),
then this operator takes the following very simple form (see e.g.~\cite{For}): for $f\in\C^{V(\G)}$ and $v\in V(\G)$,
\[
(\Delta_\G^\varphi f)(v)=\frac{1}{\mu_v}\sum_{e=(v,w)}\nu_e\,\left(f(v)-f(w)\varphi(e)\right).
\]
Finally, if the surface $\SI$ is endowed with a Riemannian metric, then a sensible choice for these weights consists in taking for $\mu_v$ the area of the star of $v$ in $\G$, and
for $\nu_e$ the quotient of the length of $e^*$ over the length of $e$. With the isoradial case in mind, this leads to the following definition.

\begin{definition}
\label{def:Delta}
Let $(\G,x)\subset\SI$ be a weighted surface graph, with weights parametrized by $x_e=\tan(\theta_e/2)$.
The associated {\em discrete Laplace operator\/} is the operator $\Delta_\G^\varphi$ on $\C^{V(\G)}$ defined by
\[
(\Delta_\G^\varphi f)(v)=\frac{1}{\mu_v}\sum_{e=(v,w)}\tan(\theta_e)\,\left(f(v)-f(w)\varphi(e)\right)
\]
for $f\in\C^{V(\G)}$ and $v\in V(\G)$, where the sum is over all oriented edges of the form $e=(v,w)$, and $\mu_v=\frac{1}{2}\sum_{e\in\EE(\G)_v}\sin(2\theta_e)$.
\end{definition}

We refer to the recent paper~\cite{Ken2} for a beautiful example of the relevance of this operator in statistical mechanics.

\subsection{The Kac-Ward operator}
\label{sub:KW}

To define the next operator, one needs to be able to measure rotation angles along curves. For planar closed curves, there is a unique sensible way to do so: one measures the rotation
angle of the velocity vector field of the curve with respect to any constant vector field on the plane. For curves embedded in an arbitrary surface, this is more tricky. As it turns out,
there is a standard geometrical tool to do this, which is called a {\em spin structure\/}. We shall not go into the trouble of recalling their formal definition (see e.g.~\cite[p.55]{Ati});
let us only mention that such a spin structure can be given by a vector field on $\SI$ with isolated zeroes of even index. Let us also recall that the group of orientation-preserving diffeomorphisms of
$\SI$ acting on the set of spin structures on $\SI$ defines two orbits: the {\em even\/} spin structures -- with so-called {\em Arf invariant\/} equal to $0$ --, and the {\em odd\/} ones, with
Arf invariant equal to $1$.
 
So, given a weighted surface graph $(\Gamma,x)\subset\SI$, let us endow $\SI$ with a Riemannian metric and let us fix a vector field $\lambda$ on $\SI$ with isolated zeroes of even index in
$\SI\setminus\G$. Also, fix a 1-cochain $\varphi\colon\EE(\G)\to\C^*$.

\begin{definition}
\label{def:KW}
The associated {\em Kac-Ward operator\/} is the operator $\KW^\varphi=\KW^\varphi(\G,x)$ on $\C^{\EE(\G)}$ defined by
\[
(\KW^\varphi f)(e)=f(e)-\varphi(e)\,x_e\sum_{\genfrac{}{}{0pt}{}{e'\in \EE(\G)_v}{e'\neq \bar{e}}}\exp\left({\textstyle\frac{i}{2}\alpha_\lambda(e,e')}\right)f(e')
\]
for $f\in\C^{\EE(\G)}$ and $e\in\EE(\G)$ with $t(e)=v$, the sum being over all $e'\in\EE(\G)$ starting where $e$ finishes, but different from $\bar{e}$.
Here, $\alpha_\lambda(e,e')$ is the rotation angle (in radians) of the velocity vector field along $e$ followed by $e'$ with respect to the vector field
$\lambda$, from $z_e$ to $z_{e'}$. (See Figure~\ref{fig:alpha}.)
\end{definition}

If the cocycle is trivial, we shall simply denote this operator by $\KW(\G,x)$, or by $\KW^\lambda(\G,x)$ if we wish to underline its dependence on the spin structure $\lambda$.

\begin{figure}[Htb]
\labellist\small\hair 2.5pt
\pinlabel {$e$} at 100 35
\pinlabel {$e'$} at 245 90
\pinlabel {$z_e$} at -20 10
\pinlabel {$z_{e'}$} at 365 133
\pinlabel {$\alpha_\lambda(e,e')$} at 310 30
\endlabellist
\begin{center}
\includegraphics[height=2cm]{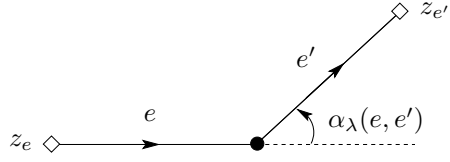}
\caption{The angle $\alpha_\lambda(e,e')$ in the planar case.}
\label{fig:alpha}
\end{center}
\end{figure}

The relevance of this operator for the study of the Ising model is given by the following facts~\cite{Cim2}.
It is well-known that the space $H^1(\SI;\{\pm 1\})$ acts freely transitively on the set of spin structures on $\SI$. Now, for any 1-cocycle $\varphi\colon\EE(\G)\to\{\pm 1\}$, the determinant
of the matrix $\KW^\varphi(\G,x)$ is the square of a polynomial in the variables $\{x_e\}_{e\in E(\G)}$ which only depends on the spin structure determined by the action of the class of
$\varphi$ on $\lambda$, and on the surface graph $\G\subset\SI$ (see subsection~\ref{sub:inv} for more details).
Finally, if $\det(\KW^\varphi(\G,x))^{1/2}$ denotes the square root with constant coefficient equal to $+1$, then the high temperature expansion of the Ising partition function on $\G$ is given by
\begin{equation}
\label{eqn:KW}
Z^\mathit{Ising}(\G,x)=\frac{1}{2^g}\sum_{\varphi\in H^1(\SI;\{\pm 1\})}(-1)^{\Arf(\varphi)}\det(\KW^\varphi(\G,x))^{1/2},
\end{equation}
where $g$ is the genus of $\SI$ and $\Arf(\varphi)\in\Z_2$ is the Arf invariant of the spin structure obtained by the action of $\varphi$ on $\lambda$.
Note that in the planar case, this equality simply reads $Z^\mathit{Ising}(\G,x)=\det(\KW(\G,x))^{1/2}$; this is the original Kac-Ward formula~\cite{KW,DZMSS}.

\subsection{The Kasteleyn operator}
\label{sub:Kast}

Recall that a {\em Kasteleyn orientation\/}~\cite{Ka1,Ka2,Ka3} on a bipartite surface graph $G\subset\SI$ can be understood as a map $\omega\colon E(G)\to\{\pm 1\}$
such that for each face $f$ of $G$,
\[
\omega(\partial f):=\prod_{e\in\partial f}\omega(e)=(-1)^{\frac{|\partial f|}{2}+1}.
\]
Two Kasteleyn orientations are called {\em equivalent\/} if they can be related by flipping the orientation of all edges adjacent to a finite number of vertices.

So, fix a bipartite weighted graph $(G,y)\subset\SI$ with vertex set $V(G)=B\cup W$, a 1-cochain $\varphi\colon\EE(G)\to\C^*$, and a Kasteleyn orientation $\omega$ on $G\subset\SI$.

\begin{definition}
\label{def:Kast}
The associated {\em Kasteleyn operator\/} is the operator $K^\varphi=K^\varphi(G,y)\colon\C^B\to\C^W$ defined by 
\[
(K^\varphi f)(w)=\sum_{e=(w,b)}\varphi(e)\,\omega(e)\,y_e\,f(b)
\]
for $f\in\C^B$ and $w\in W$, the sum being over all oriented edges of $G$ of the form $e=(w,b)$.
\end{definition}

If the cocycle is trivial, we shall simply denote this operator by $K(G,y)$, or by $K^\omega(G,y)$ if we wish to underline its dependence on the Kasteleyn orientation $\omega$.

The relevance of this operator for the study of the dimer model is given by the following formula~\cite{Tes,G-L,C-RI}. The space $H^1(\SI;\{\pm 1\})$ acts freely transitively on the set
of equivalence classes of Kasteleyn orientations on $G\subset\SI$, and this set is in equivariant correspondence with the set of spin structures on $\SI$~\cite{C-RI}. Then, the partition function of
the dimer model on $G$ is given by the formula
\begin{equation}
\label{eqn:Pf}
Z^\mathit{dimer}(G,y)=\frac{1}{2^{g}}\sum_{\varphi\in H^1(\SI;\{\pm 1\})}(-1)^{\Arf(\varphi)}\det(K^\varphi(G,y)),
\end{equation}
where $\Arf(\varphi)\in\Z_2$ denotes the Arf invariant of the spin structure obtained by the action of $\varphi$ on the spin structure corresponding to $\omega$.
Note that in the planar case, this equality simply reads $Z^\mathit{dimer}(G,y)=\det(K(G,y))$; this is Kasteleyn's celebrated theorem~\cite{Ka3}.

Note that since the spaces $\C^B$ and $\C^W$ are not identical, the determinant of the Kasteleyn operators is {\em a priori\/} not well-defined.
However, we will only be interested in the case where the number of white and black vertices of $G$ are equal. (Otherwise, both sides of the equation displayed above vanish.)
In this case, the determinants are well-defined up to a global sign. Since they also change sign when the Kasteleyn orientation is replaced by an equivalent one, the formula above
should really be understood as holding up to a global sign, or for a good choice of a Kasteleyn orientation. (One can also simply use absolute values on the right hand side.)

\subsection{The discrete Dirac operator}
\label{sub:Dirac}

Dirac operators can be defined on surfaces endowed with a complex structure and with a hermitian metric. Let us therefore consider a bipartite surface graph $G\subset\SI$ isoradially embedded in a flat
surface $\SI$, together with a fixed nowhere vanishing vector field $X$ at the vertices of $G$. (See~\cite[Section 2.1]{Cim1} for a motivation of these assumptions.)
Given an oriented edge $e=(v,v')$ of $G$, we shall write $\vartheta_X(e)$ for the oriented angle between $X(v)$ and $e$, as illustrated below.
Finally, let us write $\mu_v=\frac{1}{2}\sum_{e\in\EE(G)_v}\sin(2\theta_e)$, and fix a 1-cochain $\varphi\colon\EE(G)\to\C^*$. 

\begin{figure}[h]
\labellist\small\hair 2.5pt
\pinlabel {$e$} at 75 110
\pinlabel {$v$} at -20 22
\pinlabel {$v'$} at 160 160
\pinlabel {$\vartheta_X(e)$} at 100 45
\pinlabel {$X(v)$} at 180 0
\endlabellist
\centerline{\psfig{file=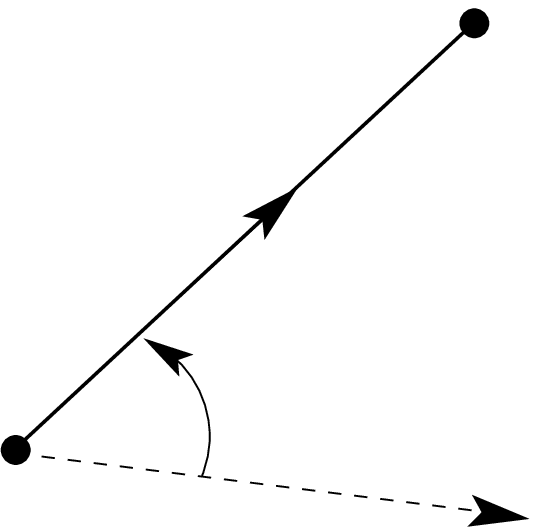,height=2.5cm}}
\label{fig:th}
\end{figure}

\begin{definition}
\label{def:Dirac}
The associated {\em discrete $\bar\partial$-operator\/} is the operator $\bar\partial^\varphi=\bar\partial^\varphi_G\colon\C^B\to\C^W$ defined by 
\[
(\bar\partial^\varphi f)(w)=\frac{1}{\mu_w}\sum_{e=(w,b)}\varphi(e)\,\exp(i\vartheta_X(e))\,\sin(\theta_e)\,f(b)
\]
for $f\in\C^B$ and $w\in W$, while the {\em discrete $\partial$-operator\/} $\partial^\varphi=\partial^\varphi_G\colon\C^W\to\C^B$ is defined by 
\[
(\partial^\varphi f)(b)=\frac{1}{\mu_b}\sum_{e=(b,w)}\varphi(e)\,\exp(-i\vartheta_X(e))\,\sin(\theta_e)\,f(w)
\]
for $f\in\C^W$ and $b\in B$. Finally the associated {\em discrete Dirac operator\/} is the operator
$\Di^\varphi=\Di^\varphi_G=\left(\begin{smallmatrix}0&-\partial^\varphi\cr \bar\partial^\varphi & 0\end{smallmatrix}\right)\colon\C^B\oplus\C^W\to \C^B\oplus\C^W$.
\end{definition}

Note that if $\SI$ is the flat plane (with $\varphi=1$ and $X$ a constant vector field), then these definitions coincide with the ones of~\cite{Ken,CS09}.
More generally, if $G$ is isoradially embedded in a flat surface and $\varphi$ is a discrete spin structure, then we get back~\cite[Definition 3.9]{Cim1}.
It should come as no surprise that a vector field is required: on Riemann surfaces, the $\partial$ and $\bar\partial$ operators take values in $(1,0)$ and $(0,1)$-forms, respectively.
In any case, a function $f\in\C^B$ being {\em discrete holomorphic\/} at $w\in W$, i.e. satisfying $(\bar\partial^\varphi f)(w)=0$, is independent of the vector field.

We will be particularly interested in two families of bipartite graphs, namely the graphs $D_\G$ and $C_\G$ associated to an arbitrary isoradial graph $\G$ as described in subsection~\ref{sub:graph}.
The corresponding Dirac operators will be denoted by $\Di^\varphi_D=\left(\begin{smallmatrix}0&-\partial_D^\varphi\cr \bar\partial_D^\varphi & 0\end{smallmatrix}\right)$ and
$\Di^\varphi_C=\left(\begin{smallmatrix}0&-\partial_C^\varphi\cr \bar\partial_C^\varphi & 0\end{smallmatrix}\right)$.

%%%%%%%%%%%%%%%%%%%%%%%%%%%%%%%%%%%

\section{Relations between operators}
\label{sec:rel}

The aim of this section is to exhibit natural relations between the operators defined in Section~\ref{sec:def}. A couple of applications are also included, but the most interesting consequences 
will be presented in Section~\ref{sec:s-holo}. More precisely, we begin by showing one of our main results: the fact that the Kac-Ward operator on an arbitrary surface graph $\G$ and the Kasteleyn
operator on the associated bipartite graph $C_\G$ can be explicitly related (subsection~\ref{sub:corr}). As an immediate consequence, we obtain a Kramers-Wannier type duality result for Kac-Ward
determinants. In subsection~\ref{sub:Kast-D}, we show that in the isoradial case, the Kasteleyn and discrete Dirac operators on $C_\G$ are conjugate.
In the next paragraph, we observe that the square of the Dirac operator on $C_\G$ is closely related to the discrete Laplace operator on an associated graph.
Finally, in subsection~\ref{sub:D-D}, we relate the discrete Dirac operators on $C_\G$ and on the double $D_\G$ of $\G$. Note that the results of these last two subsections will not be referred to
in the rest of the paper.

\subsection{Kac-Ward versus Kasteleyn operators}
\label{sub:corr}

In this paragraph, we relate the Kac-Ward operator on an arbitrary weighted surface graph $(\G,x)\subset\SI$ to the Kasteleyn operator on the associated {\em bipartite\/} graph
$(C_\G,y(x))\subset\SI$ (Theorem~\ref{thm:corr}). As an immediate consequence, we obtain an equality between their respective determinants (Corollary~\ref{cor:det}). Note that such a Kac-Ward determinant
is known to be equal to the Kasteleyn determinant of the associated {\em Fisher graph\/} (see~\cite[Proposition 4.6]{Cim2}). However, Fisher graphs are never bipartite, thus
not allowing to use the full power of the dimer model theory. In the present case, not only is $C_\G$ bipartite, but Corollary~\ref{cor:det} immediately implies a generalized Kramers-Wannier
duality for Kac-Ward determinants on arbitrary weighted surface graphs (Corollaries~\ref{cor:KW1} and~\ref{cor:KW2}). These statements simultaneously generalize~\cite[Theorem 4.4]{Cim3}, which deals with
the critical isoradial case, and~\cite[Corollaries 3.3 and 3.4]{C-D}, which deal with the toric case.

\begin{figure}[Htb]
\labellist\small\hair 2.5pt
\pinlabel {$e$} at 282 25
\pinlabel {$R(e)$} at 340 265
\pinlabel {$R^2(e)$} at 210 375
\pinlabel {$R^3(e)$} at 70 370
\pinlabel {$\vdots$} at 0 200
\pinlabel {$R^{-1}(e)$} at 100 17
\pinlabel {$\beta_e$} at 220 170
\endlabellist
\begin{center}
\includegraphics[height=5cm]{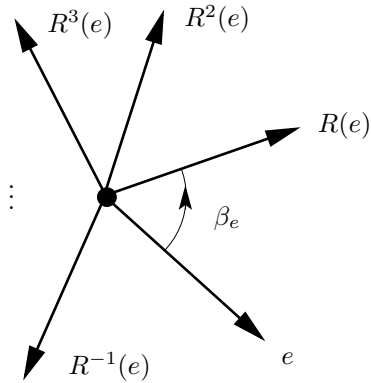}
\caption{The rotated edges, and the angle $\beta_e$.}
\label{fig:R}
\end{center}
\end{figure}

To state the main result of this section, we need to introduce several notations. Let $(\G,x)\subset\SI$ be an arbitrary weighted surface graph, $\SI$ being endowed with a Riemannian metric,
and assume that the edges of $\G$ are smoothly embedded in $\SI$. Also, let us fix a vector field $\lambda$ on $\SI$ with zeroes of even index in $\SI\setminus\G$.
For any $e\in\EE(\G)=:\EE$, write $D_e=\exp(ia_e)$, where $a_e$ denotes the oriented angle at $z_e$ between $\lambda$ and $e$. (The letters $D$ and $a$ stand for ``direction'' and ``argument'',
respectively.) Note the equality $D_{\bar{e}}=-D_e$.
Given a vertex $v\in V(\G)$, let us cyclically order the elements of $\EE_v$ by turning counterclockwise around $v$. (As $\SI$ is orientable, this can be done in a consistent way.)
For $e\in\EE_v$, let $R(e)\in\EE_v$ denote the next edge with respect to this cyclic order, and set $\beta_e=\pi+\alpha_\lambda(\bar{e},R(e))$. By definition of $\alpha_\lambda$
(recall Definition~\ref{def:KW}), $\beta_e$ is the rotation angle between $e$ and $R(e)$ measured with respect to the fixed vector field $\lambda$ (see Figure~\ref{fig:R}).
Note the equality $D_{R(e)}=\exp(i\beta_e)D_e$. Finally, set $q_e=\exp(\frac{i}{2}\beta_e)$.

This allows us to define automorphisms $J,R,D$ and $q$ of the space $\C^\EE$: $J$ is simply given by $(Jf)(e)=f(\bar{e})$, while $R$ is defined by $(Rf)(e)=f(R(e))$.
Finally, we shall write $D$ for the diagonal automorphism of $\C^\EE$ given by $(D f)(e)=D_e\,f(e)$, and similarly for $q$, for any weight system, and for $\varphi$.
To conclude this long list of notations, let $\psi_B\colon\EE\to B:=B(C_\G)$ (resp. $\psi_W\colon\EE\to W:=W(C_\G)$) denote the bijection mapping each oriented edge $e$ of $\G$ to the unique black
(resp. white) vertex of $C_\G$ immediately to the right (resp. left) of $e$, as illustrated in Figure~\ref{fig:psi}. These bijections induce isomorphisms $\psi_B\colon\C^B\to\C^\EE$ and
$\psi_W\colon\C^W\to\C^\EE$.

\begin{figure}[Htb]
\labellist\small\hair 2.5pt
\pinlabel {$e=\psi_W^{-1}(w)$} at 600 35
\pinlabel {$R(e)$} at -25 200
\pinlabel {$\psi_B(R(e))$} at 150 180
\pinlabel {$\psi_B(e)$} at 200 0
\pinlabel {$\psi_B(\bar{e})$} at 420 140
\pinlabel {$w$} at 171 131
\endlabellist
\begin{center}
\includegraphics[height=2.3cm]{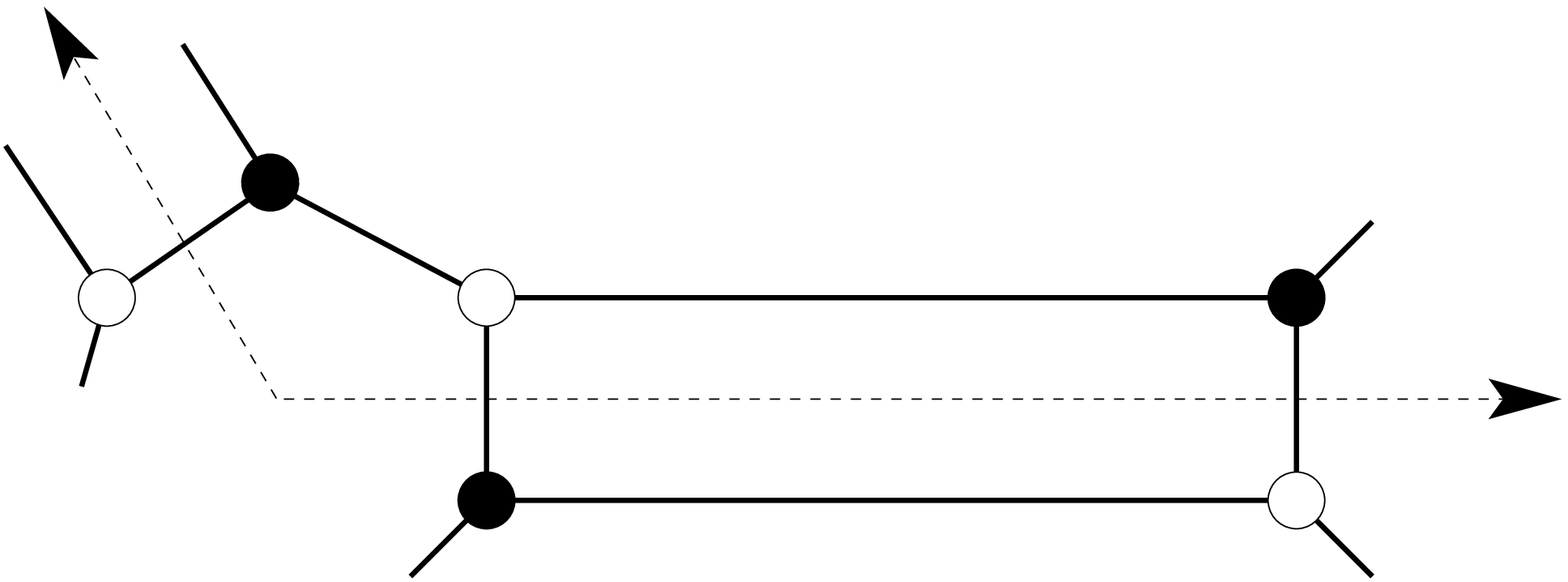}
\caption{The bijections $\psi_B\colon\EE\to B$ and $\psi_W\colon\EE\to W$.}
\label{fig:psi}
\end{center}
\end{figure}

\begin{theorem}
\label{thm:corr}
For any 1-cochain $\varphi\colon\EE\to\C^*$, the following diagram commutes:
\[
\xymatrix{
\C^\EE\ar[rr]^{\KW^\varphi(\G,x)} & \hspace{1cm} & \C^\EE\\
\C^\EE\ar[u]^{I-qR}_\simeq & \hspace{1cm} & \C^\EE\ar[u]_{I-i\varphi xJ}^\simeq \\
\C^\EE\ar[u]^{D^{-1/2}}_\simeq & \hspace{1cm} & \C^\EE\ar[u]_{D^{-1/2}}^\simeq \\
\C^B\ar[u]^{\psi_B}_\simeq \ar[rr]^{K^{\varphi_C}(C_\G,y)} & \hspace{1cm} & \C^W.\ar[u]_{\psi_W}^\simeq}
\]
\end{theorem}

Before proving this theorem, let us make a couple of remarks.
First, note that the choice of the square root of the operator $D$ has no importance: changing the sign of the square root of the coefficient $D_e$ (say, on the left-hand side of the diagram)
simply amounts to inverting the Kasteleyn orientation of all edges adjacent to $\psi_B(e)$, i.e. replacing it with an equivalent Kasteleyn orientation.
Also, the dependence of this diagram on the equivalence class of the Kasteleyn orientation is well-understood. Such equivalence classes are in natural one-to-one correspondence
with spin structures on $\SI$ (see~\cite[Section 4]{C-RI}), which can be described by vector fields with zeroes of even index. Moreover, the paper~\cite{C-RI} contains an explicit construction
of such a vector field $\lambda_\omega$ from a Kasteleyn orientation $\omega$, and one can check that the diagram above commutes for such an orientation $\omega$ and vector field $\lambda_\omega$.

\begin{proof}[Proof of Theorem~\ref{thm:corr}]
We shall start with the most technical part of the demonstration, i.e. the computation of the composition $(I-i\varphi xJ)^{-1}\KW^\varphi(I-qR)$. To achieve this goal, it will be convenient to
consider the operator $\mathit{Succ}\in\mathit{End}(\C^\EE)$ defined as follows: if $e$ is an oriented edge with terminus $t(e)=v$, then
\[
\mathit{Succ}(f)(e)=\varphi(e)x_e\sum_{e'\in \EE_v}\omega(e,e')\,f(e'),
\]
where $\omega(e,e')=\exp(\frac{i}{2}\alpha_\lambda(e,e'))$ for $e'\neq \bar{e}\in \EE_v$ and $\omega(e,\bar{e})=-i$.
Also, let $T\in\mathit{End}(\C^\EE)$ be the endomorphism given by $T=\mathit{Succ}+i\varphi xJ$. By definition, the Kac-Ward operator $\KW^\varphi$ is equal to
$I-T$. Now, the operator $(I+i\varphi xJ)(I-T)$ decomposes into
\[
(I+i\varphi xJ)(I-T)=I-\mathit{Succ}+\mathit{Com},
\]
where
\[
\mathit{Com}(f)(e)=(-i\varphi xJT)(f)(e)=-i\varphi(e) x_e(Tf)(\bar{e})=-i x_e^2\sum_{\genfrac{}{}{0pt}{}{e'\in \EE_v}{e'\neq e}}\omega(\bar{e},e')\,f(e')
\]
if $e$ has origin $o(e)=v$. Let us compute the composition $(I+i\varphi xJ)(I-T)(I-qR)$ using the decomposition displayed above. If $e$ has terminus $t(e)=v$, then
\begin{align*}
(\mathit{Succ}(I-qR))(f)(e)&=\mathit{Succ}(f)(e)-\mathit{Succ}(q(R(f)))(e)\cr
	&=\varphi(e) x_e\sum_{e'\in \EE_v}\omega(e,e')\left(f(e')-q_{e'}f(R(e'))\right)\cr
	&=\varphi(e) x_e\sum_{e'\in \EE_v}\left(\omega(e,e')-\omega(e,R^{-1}(e'))\,q_{R^{-1}(e')}\right)f(e')\cr
	&=-2i\varphi(e) x_e\,f(\bar{e}).
\end{align*}
Therefore, we have the equality $\mathit{Succ}(I-qR)=-2i\varphi xJ$. Similarly, given $e$ with origin $o(e)=v$,
\begin{align*}
(\mathit{Com}(I-qR))(f)(e)&=\mathit{Com}(f)(e)-\mathit{Com}(q(R(f)))(e)\cr\cr
	&=-ix_e^2\sum_{e'\in \EE_v\setminus\{e\}}\omega(\bar{e},e')\left(f(e')-q_{e'}f(R(e'))\right)\cr
	&=-ix_e^2\hskip-.2cm\sum_{e'\in \EE_v\setminus\{e,R(e)\}}\hskip-.2cm\left(\omega(\bar{e},e')-\omega(\bar{e},R^{-1}(e'))\,q_{R^{-1}(e')}\right)f(e')\cr
	&\phantom{=}+ix_e^2\left(\omega(\bar{e},R^{-1}(e))\,q_{R^{-1}(e)}\,f(e)-\omega(\bar{e},R(e))\,f(R(e))\right)\cr
	&=-x^2_e(I+qR)(f)(e).
\end{align*}
These two equalities lead to
\begin{align*}
(I+i\varphi xJ)(I-T)(I-qR)&=(I-\mathit{Succ}+\mathit{Com})(I-qR)\cr
	&=(I-qR)+2i\varphi xJ-x^2(I+qR)\cr
	&=(1-x^2)+2xi\varphi J-(1+x^2)qR.
\end{align*}
This allows us to determine the operator $M^\varphi$ defined by the commuting diagram
\[
\xymatrix{
\C^\EE\ar[rr]^{\KW^\varphi(\G,x)}& \hspace{1cm} & \C^\EE\\
\C^\EE\ar[u]^{I-qR}_\simeq\ar[rr]^{M^\varphi} & \hspace{1cm} & \C^\EE\ar[u]_{I-i\varphi xJ}^\simeq .}
\]
It is given by
\begin{align*}
M^\varphi&=(I-i\varphi xJ)^{-1}(I-T)(I-qR)\cr
	&=\frac{1}{1+x^2}(I+i\varphi xJ)(I-T)(I-qR)\cr
	&=\frac{1-x^2}{1+x^2}+\frac{2x}{1+x^2}i\varphi J-qR\cr
	&=\cos(\theta)+\sin(\theta)i\varphi J-qR,
\end{align*}
using the parametrization $x=\tan(\theta/2)$ of the weights.

Now, observe that the three maps $\psi_B\circ\psi_W^{-1}$, $\psi_B\circ\bar{\phantom{e}}\circ \psi_W^{-1}$ and $\psi_B\circ R\circ\psi_W^{-1}$ associate to a fixed $w\in W$ the three black vertices
of $C_\G$ adjacent to $w$, as demonstrated in Figure~\ref{fig:psi}. Therefore, the operator $\widetilde{K}^\varphi$ defined by the commuting diagram
\[
\xymatrix{
\C^\EE\ar[rr]^{M^\varphi} & \hspace{1cm} & \C^\EE\\
\C^B\ar[u]^{\psi_B}_\simeq \ar[rr]^{\widetilde{K}^\varphi} & \hspace{1cm} & \C^W\ar[u]_{\psi_W}^\simeq}
\]
is given by the coefficients
\[
\widetilde{K}^\varphi_{wb}=
\begin{cases}
\cos(\theta_e)&\text{if the edge $(w,b)$ is perpendicular to $e\in\EE$;}\\ 
\varphi(e) i \sin(\theta_e)&\text{if $(w,b)$ is to the left of $e\in\EE$;} \\
-q_e&\text{if $(w,b)$ is in the ``corner" of $e$ and $R(e)$;} \\
0 & \text{if $w$ and $b$ are not adjacent in $C_\G$.}
\end{cases}
\]
This is illustrated below.

\begin{figure}[h]
\labellist\small\hair 2.5pt
\pinlabel {$e$} at 610 75
\pinlabel {$R(e)$} at -40 260
\pinlabel {$\varphi(e) i \sin(\theta_e)$} at 340 140
\pinlabel {$\varphi(e)^{-1} i \sin(\theta_e)$} at 340 10
\pinlabel {$\cos(\theta_e)$} at 95 40
\pinlabel {$\cos(\theta_e)$} at 590 15
\pinlabel {$-q_e$} at 175 160
\endlabellist\begin{center}
\includegraphics[height=3cm]{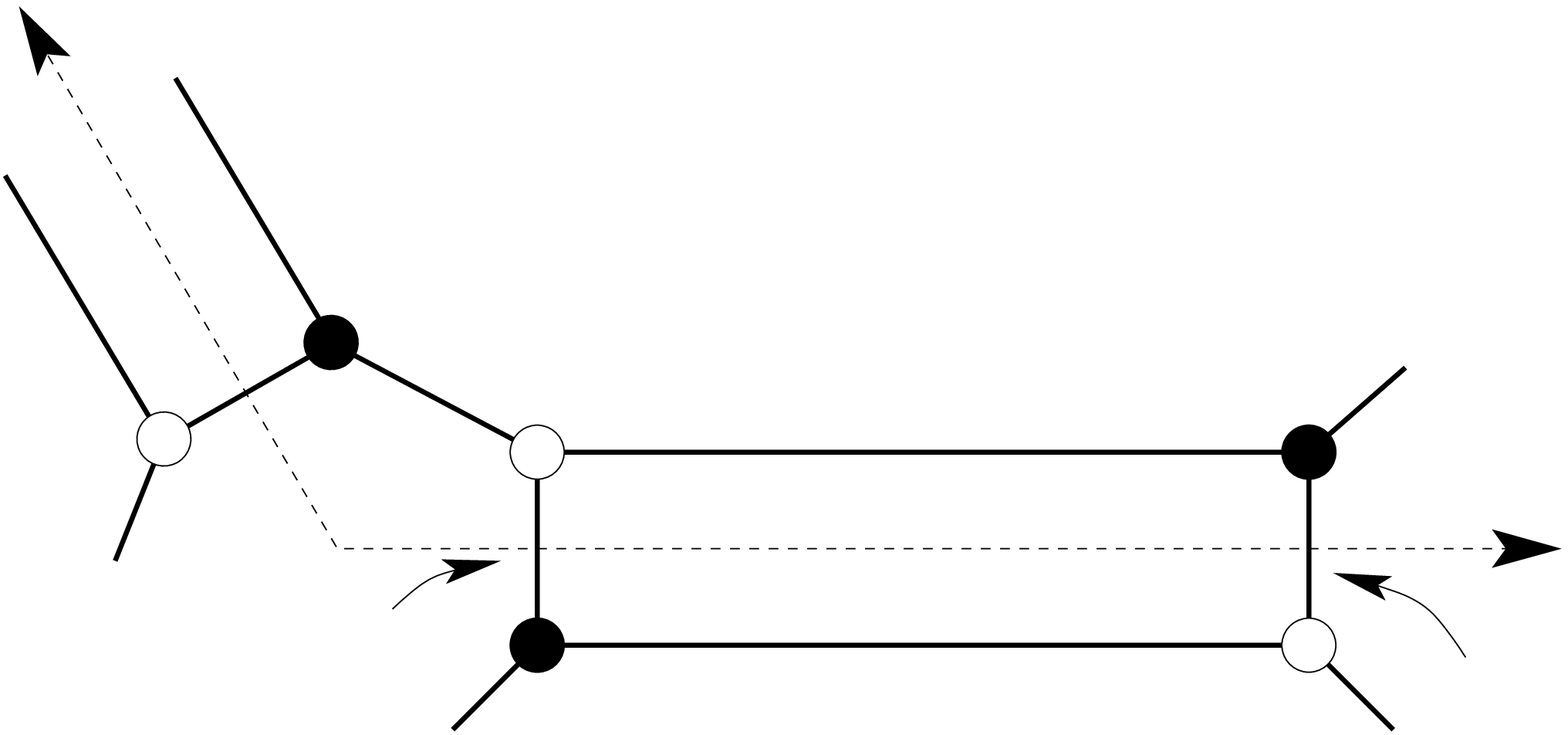}
\end{center}
\end{figure}

\noindent In other words, $\widetilde{K}^\varphi$ is precisely the Kasteleyn operator of $(C_\G,y)$ associated to the 1-cocycle $\varphi_C\colon\EE(C_\G)\to\C^*$ (recall Figure~\ref{fig:cocycles})
and to the map $\tilde\omega\colon E(C_\G)\to S^1\subset\C^*$ given by
\[
\tilde\omega(w,b)=
\begin{cases}
1&\text{if $(w,b)$ is perpendicular to $e\in\EE$;}\\ 
i&\text{if $(w,b)$ is to the left of $e\in\EE$;} \\
-q_e&\text{if $(w,b)$ is in the corner of $e$ and $R(e)$.}
\end{cases}
\]
Extend this map to a 1-cochain $\tilde\omega\colon\EE(C_\G)\to S^1$ by setting $\tilde\omega(b,w)=\tilde\omega(w,b)^{-1}$. Now, observe that for any face
$f$ of $C_\G\subset\SI$,
\[
\tilde\omega(\partial f)=\prod_{e\in\partial f}\tilde\omega(e)=(-1)^{\frac{|\partial f|}{2}+1}.
\]
This is obvious for the rectangular faces; for the faces corresponding to vertices of $\G$, use the fact that the angles $\beta_e$ add up to $2\pi$ around each vertex;
for the faces corresponding to faces of $\G$, use the fact that the angles $\alpha_\lambda(e,e')$ add up to an odd multiple of $2\pi$ around each face since the vector field $\lambda$ has
zeroes of even index. Furthermore, if $\gamma$ denotes a 1-cycle in $C_\G\subset\SI$, then one easily checks that $\tilde\omega(\gamma)^2=\exp(i\,\mathit{rot}_\lambda(\gamma))=1$,
where $\mathit{rot}_\lambda(\gamma)$ is the rotation angle of the velocity vector field along the closed curve $\gamma$ with respect to $\lambda$. Therefore, the $S^1$-valued 1-cochain
$\tilde\omega$ is cohomologous to a
Kasteleyn orientation $\omega$. In other words, $\tilde\omega$ can be transformed into $\omega$ by a sequence of the following transformation: multiply all the edges adjacent to a fixed vertex
of $C_\G$ by some complex number of modulus 1. This defines diagonal operators $d_B$ and $d_W$ such that the following diagram commutes:
\[
\xymatrix{
\C^B\ar[rr]^{\widetilde{K}^\varphi} & \hspace{1cm} & \C^W\\
\C^B\ar[u]^{d_B}_\simeq \ar[rr]^{K^{\varphi_C}(C_\G,y)} & \hspace{1cm} & \C^W.\ar[u]_{d_W}^\simeq}
\]
By connectedness of $\G$, these diagonal operators are unique up to simultaneous multiplication by a fixed element of $S^1$.

To complete the proof of the theorem, it remains to check that
\[
\psi_B\circ d_B=D^{-1/2}\circ\psi_B\quad\text{and}\quad\psi_W\circ d_W=D^{-1/2}\circ\psi_W
\]
for some well-chosen square root of the operator $D$. To do so, let us make the diagonal operators $d_B=(d_b)_{b\in B}$ and $d_W=(d_w)_{w\in W}$ more explicit.
Focusing on the rectangular face of $C_\G$ around the edge $e$ leads to the equations $d_{\psi_W(e)}=\pm d_{\psi_B(e)}$ and $d_{\psi_W(e)}=\pm i\,d_{\psi_B(\bar{e})}$.
Similarly, looking at the face of $C_\G$ around the vertex $v$ leads to the equations $d_{\psi_B(R(e))}=\pm q^{-1}_e d_{\psi_W(e)}$ for all $e\in\EE_v$.
In a nutshell, the operators $d_B^2$ and $d_W^2$ are uniquely determined (up to a global rotation) by the equations
\[
d^2_{\psi_B(e)}=d^2_{\psi_W(e)}=-d^2_{\psi_B(\bar{e})}\quad\text{and}\quad d^2_{\psi_B(R(e))}=\exp(-i\beta_e)d^2_{\psi_B(e)}\quad\text{for all $e\in\EE$}.
\]
By definition, the operator $D$ satisfies the equations $D_e=-D_{\bar{e}}$ and $D_{R(e)}=\exp(i\beta_e)D_e$. It follows that $d_B,d_W$ can be chosen so that
$d^2_{\psi_B(e)}=d^2_{\psi_W(e)}=D_e^{-1}$ for all $e\in\EE$. This is equivalent to the equalities $\psi_\bullet\circ d_\bullet=D^{-1/2}\circ\psi_\bullet$ for $\bullet=B,W$, and the proof is completed.
\end{proof}

\bigbreak

Theorem~\ref{thm:corr} easily leads to the following result, which was first obtained in the genus one case by Duminil-Copin and the author~\cite{C-D}.
(On how to make sense of the determinant of the Kasteleyn operator, recall the remark at the end of paragraph~\ref{sub:Kast}.)

\begin{corollary}
\label{cor:det}
Given a weighted surface graph $(\G,x)\subset\SI$, there is a Kasteleyn orientation on $C_\G\subset\SI$ such that
\[
\det(\KW^\varphi(\G,x))= 2^{-|V(\G)|}\prod_{e\in E(\G)}(1+x_e^2)\,\det(K^{\varphi_C}(C_\G,y))
\]
for any 1-cochain $\varphi\colon\EE\to\C^*$.
\end{corollary}
\begin{proof}
First note that 
\[
\det(I-i\varphi xJ)=\prod_{e\in E(\G)}\det\left(\begin{array}{cc}1&-i\varphi(\bar{e}) x_e\cr -i\varphi(e) x_e&1\cr\end{array}\right)=\prod_{e\in E(\G)}(1+ x_e^2).
\]
Also, the partition ${\EE}=\bigsqcup_{v\in V(\G)}\EE_v$ induces a decomposition $I-qR=\bigoplus_{v\in V(\G)}(I-qR)_v$. This leads to 
\[
\det(I-qR)=\prod_{v\in V(\G)}\det((I-qR)_v)=\prod_{v\in V(\G)}\left(1-{\textstyle\prod_{e\in \EE_v}q_e}\right)=2^{|V(\G)|},
\]
since $\sum_{e\in\EE_v}\beta_e=2\pi$. Hence, Theorem~\ref{thm:corr} implies that the equality of the corollary holds for any 1-cochain $\varphi$ up to multiplication by a fixed complex number of modulus 1.
For $\varphi$ taking values in $\{\pm 1\}$, both sides are real; therefore, the equality holds up to a global sign. One can then choose a Kasteleyn orientation such that the identity holds.
\end{proof}

Let us now assume that $\varphi\colon\EE\to\C^*$ is a 1-cocycle. In this case, one easily checks that the determinant of $\KW^\varphi(\G,x)$ only depends on the cohomology class of $\varphi$ in
$H^1(\SI;\C^*)$. By abuse of notation, we shall also denote such a class by $\varphi$. Here is the announced generalized Kramers-Wannier duality.

\begin{corollary}
\label{cor:KW1}
For any weighted graph $(\G,x)\subset\SI$ and any class $\varphi\in H^1(\SI;\C^*)$,
\[
2^{|V(\G)|}\prod_{e\in E(\G)}(1+x_e)^{-1} \det(\KW^\varphi(\G,x))=2^{|V(\G^*)|}\prod_{e\in E(\G)}(1+x^*_e)^{-1} \det(\KW^\varphi(\G^*,x^*)).
\]
\end{corollary}
\begin{proof}
As mentioned in subsection~\ref{sub:graph}, the weighted graph $(C_{\G^*},y(x^*))$ associated to $(\G^*,x^*)$ is equal to the weighted graph $(C_\G,y(x))$ associated to $(\G,x)$.
Fix a cocycle $\varphi_D$ on $D_\G$ as in Figure~\ref{fig:cocycles} with $\varphi_1\varphi_3^{-1}=\varphi_2\varphi_4^{-1}$ -- any class can be represented by such a cocycle --
and consider the associated cocycles $\varphi$, $\varphi_*$, $\varphi_C$ and $(\varphi_*)_C$, who all define the same cohomology class.
The cocycles $\varphi_C$ and $(\varphi_*)_C$ are not equal, but cohomologous, and the corresponding transformations of the Kasteleyn matrices do not change their determinant. 
Corollary~\ref{cor:det} together with the equality $\frac{1+x^2}{1+x}=\frac{1+(x^*)^2}{1+x^*}$ then gives the result.
\end{proof}

As mentioned in subsection~\ref{sub:KW}, for a cocycle $\varphi\colon\EE(\G)\to\{\pm 1\}$, $\det(\KW^\varphi(\G,x))$ is the square of a polynomial in the weight variables $x_e$.
As the constant coefficient of this determinant is equal to $1$, we can pick such a square root $|\KW^\varphi(\G,x)|^{1/2}$ by requiring its constant coefficient to be $+1$.
Taking a closer look at the sign leads to the following duality.

\begin{corollary}
\label{cor:KW2}
For any weighted graph $(\G,x)\subset\SI$ and any $\varphi\in H^1(\SI;\{\pm 1\})$,
\[
2^{\frac{|V(\G)|}{2}}\prod_{e\in E(\G)}(1+x_e)^{-\frac{1}{2}}|\KW^\varphi(\G,x)|^{\frac{1}{2}}=
(-1)^{A(\varphi)}\,2^{\frac{|V(\G^*)|}{2}}\prod_{e\in E(\G)}(1+x^*_e)^{-\frac{1}{2}} |\KW^\varphi(\G^*,x^*)|^{\frac{1}{2}},
\]
where $A(\varphi)\in\Z_2$ is the Arf invariant of the spin structure obtained by the action of $\varphi$ on $\lambda$.
\end{corollary}
\begin{proof}
By Corollary~\ref{cor:KW1}, we only need to determine the sign $A(\varphi)$ in the equation above. Let $g$ denote the genus of $\SI$. Setting $x=1$ (and therefore, $x^*=0$) leads to
\[
|\KW^\varphi(\G,1)|^{1/2}=(-1)^{A(\varphi)}2^{(|V(\G^*)|+|E(\G)|-|V(\G)|)/2}=(-1)^{A(\varphi)}2^{|V(\G^*)|+g-1},
\]
using the fact that $|V(G)|-|E(G)|+|V(G^*)|$ is equal to the Euler characteristic $\chi(\SI)=2-2g$.
Furthermore, the Ising partition function $Z^\mathit{Ising}(\G,x)$ with weights $x=1$ is nothing but the cardinality of the $\Z_2$-vector space of 1-cycles modulo 2 in $\G$.
Since $\G$ is connected, the dimension of this space is classically equal to $|E(\G)|-|V(\G)|+1$. The Kac-Ward formula~\eqref{eqn:KW} now implies
\[
2^{|E(\G)|-|V(\G)|+1}=Z^\mathit{Ising}(\G,1)=\frac{1}{2^g}\sum_{\varphi\in H^1(\SI;\{\pm 1\})}(-1)^{\Arf(\varphi)}(-1)^{A(\varphi)}2^{|V(G^*)|+g-1},
\]
that is,
\[
\sum_{\varphi\in H^1(\SI;\{\pm 1\})}(-1)^{\Arf(\varphi)+A(\varphi)}=2^{2-\chi(\SI)}=2^{2g}.
\]
Since the set $H^1(\SI;\{\pm 1\})$ has precisely $2^{2g}$ elements, this shows that $A(\varphi)=\Arf(\varphi)$ for all $\varphi$, and the corollary is proved.
\end{proof}

To illustrate the power of Theorem~\ref{thm:corr}, let us finally show how it allows us to easily transfer the recent results Kenyon-Sun-Wilson~\cite{KSW} from the dimer to the
Ising model. Consider a locally finite planar graph $\mathcal{G}$ with edge weights $J=(J_e)_e\in(0,\infty)^{E(\mathcal{G})}$
and faces that are bounded topological discs. Let us assume that this weighted graph is {\em biperiodic\/}, i.e. invariant under the action of a lattice $L\simeq\Z^2$.
For $n\ge 1$, let $\G_n$ denote the finite weighted toric graph given by
$\G_n=\mathcal{G}/n L\subset\mathbb{T}^2$, and set $\G_1=:\G$. Fixing a basis of $L$, one can identify $H^1(\mathbb{T}^2;\C^*)$ with $(\C^*)^2$ so that any 1-cocycle
$\varphi\colon\mathbb{E}\to\C^*$ defines a pair $(z,w)$ of non-vanishing complex numbers. Finally, let us write $P_n(z,w)=\det\KW^\varphi(\G_n,x)$, where $x_e=\tanh(\beta J_e)$,
set $P_1=:P$, and denote by $H=\left(\begin{smallmatrix}A_z&B\\ B&A_w\end{smallmatrix}\right)$ the Hessian of $P$ at $(1,1)$.

\begin{corollary}\label{cor:tau}
At the critical inverse temperature $\beta=\beta_c$, the Ising partition function on $\G_n$ satisfies
\[
\log Z^J_{\beta_c}(\G_n)= n^2 f(\beta_c) + \mathsf{fsc}_1(\tau)+o(1),
\]
where $f$ is the free energy per fundamental domain
\[
f(\beta)=|V(\G)|\log(2)+\sum_{e\in E(\G)}\log\cosh(\beta J_e)+\frac{1}{2(2\pi i)^2} \int_{\mathbb T^2} \log P(z,w)\frac{dz}{z}\frac{dw}{w},
\]
$\tau$ is the modular parameter given by
\[
\tau=\frac{-B+i\sqrt{A_zA_w-B^2}}{A_w},
\]
and $\mathsf{fsc}_1(\tau)$ is the explicit universal finite-size correction term given in~\cite[Theorem 2 (a)]{KSW}, which is invariant under modular transformations.
\end{corollary}
\begin{proof}
Given a cocycle $\varphi\colon\mathbb{E}(\G)\to\C^*$ representing a class $(z,w)\in(\C^*)^2$, let us write $Q(z,w)=\det K^{\varphi_C}(C_{\G},y)$. (Recall that the cocycles
$\varphi$ and $\varphi_C$ define the same cohomology class.) Corollary~\ref{cor:det} now reads
\[
P(z,w)=2^{-|V(\G)|}\prod_{e\in E(\G)}(1+x_e^2)\,\,Q(z,w),
\]
and implies that the curves defined by $P(z,w)$ and $Q(z,w)$ coincide. Since $C_\G$ is bipartite, one can apply results of Kenyon-Okounkov-Sheffield~\cite{KOS}
and conclude that this is a special Harnack curve. As demonstrated in~\cite{C-D}, it allows us to show that $P(z,w)$ is strictly positive on $\mathbb{T}^2\setminus\{(1,1)\}$,
and vanishes at $(1,1)$ if and only if the inverse temperature is critical.
Recall that the high temperature expansion~\cite{vdW} for the Ising model on $(\G_n,J)$ can be stated as
\[
Z^J_\beta(\G_n)=2^{n^2|V(\G)|}\Big(\prod_{e\in E(\G)}\cosh(\beta J_e)\Big)^{n^2}\,Z^\mathit{Ising}(\G_n,x),
\]
while the genus one Kac-Ward formula~(\ref{eqn:KW}) reads
\[
Z^\mathit{Ising}(\G_n,x)=\frac{1}{2}\left(-P_n(1,1)^{1/2}+P_n(-1,1)^{1/2}+P_n(1,-1)^{1/2}+P_n(-1,-1)^{1/2}\right).
\]
The conclusion now follows from~\cite[Theorem 1]{KSW} in the special case $E=\left(\begin{smallmatrix}n&0\\ 0&n\end{smallmatrix}\right)$, applied to the equation above.
\end{proof}

\begin{example}
Consider the case of the rectangular lattice $\mathcal{G}$ with horizontal (resp. vertical) edges having weight $J$ (resp. $K$). Fixing the natural fundamental domain and basis of $\Z^2$, we obtain
\[
P(z,w)=(1+x^2)(1+y^2)-x(1-y^2)(z+z^{-1})-y(1-x^2)(w+w^{-1}),
\]
where $x=\tanh(\beta J)$ and $y=\tanh(\beta K)$. This leads to
\[
H=\begin{pmatrix}-2x(1-y^2)&0\cr 0&-2y(1-x^2)\end{pmatrix}\text{ and } \tau=i\sqrt{\frac{y(1-x^2)}{x(1-y^2)}}.
\]
Now, the model is at the critical temperature if and only if $x$ and $y$ satisfy the equality $1=x+y+xy$. Using the parametrization $x=\tan(\theta/2)$ and $y=\tan(\rho/2)$,
this equation can be written as $\theta+\rho=\frac{\pi}{2}$. This leads to the modular parameter $\tau=i\tan(\theta)$, which defines the isoradial embedding of $\mathcal{G}$ whose critical weights are the ones we started with.
\end{example}

\subsection{Kasteleyn versus discrete Dirac operators}
\label{sub:Kast-D}

Let us now assume that the graph $\G$ is isoradially embedded in a flat surface $\SI$, and endowed with the corresponding critical weights $x_e=\tan(\theta_e/2)$ (recall paragraph~\ref{sub:iso}).
Then, the associated graph $C_\G$ is also isoradially embedded in $\SI$, so one can define a discrete $\bar{\partial}$-operator $\bar{\partial}_C^\varphi\colon\C^B\to\C^W$ as explained
in paragraph~\ref{sub:Dirac}. Recall that this discretization is particularly relevant when the 1-cochain $\varphi$ is a so-called {\em discrete spin structure\/}, that is, when it represents
one of the $2^{2g}$ inverse square roots of the holonomy of the flat metric on $\SI$. Recall also that this operator $\bar{\partial}_C^\varphi$ depends on the choice of a vector field $X$ at the white
vertices of $C_\G$; without loss of generality, we shall assume that this vector field at the vertex $w=\psi_W(e)$ points towards $\psi_B(e)\in B$ (recall Figure~\ref{fig:psi}).

The main result of this paragraph is that the Kasteleyn operators associated to $2^{2g}$ non-equivalent Kasteleyn orientations on $(C_\G,y)$ are
simultaneously conjugate to the discrete $\bar{\partial}$-operators associated to $2^{2g}$ non-cohomologous discrete spin structures.
This can be understood as an instructive and explicit special case of~\cite[Theorem 4.8]{Cim1}, which deals with general bipartite isoradial graphs.
(See also~\cite[Theorem 10.1]{Ken} for the planar case). 

\begin{proposition}
\label{prop:Kast-D}
Given any Kasteleyn orientation $\omega$ on $C_\G\subset\SI$, there exists a unique discrete spin structure $\varphi_\omega\colon\EE(C_\G)\to S^1$ such that the following diagram commutes
\[
\xymatrix{
\C^B\ar[rr]^{K^\omega(C_\G,y)} & \hspace{1cm} & \C^W\\
\C^B\ar[u]^{\exp\left(-\frac{i}{2}\theta_B\right)}_\simeq\ar[rr]^{\bar{\partial}_C^{\varphi_\omega}}& \hspace{1cm} & \C^W,\ar[u]_{\exp\left(-\frac{i}{2}\theta_W\right)\circ\mu_W}^\simeq}
\]
where the diagonal operators $\theta_B=(\theta_b)_{b\in B}$, $\theta_W=(\theta_w)_{w\in B}$ and $\mu_W=(\mu_w)_{w\in W}$ are defined by $\theta_b=\theta_e$ if $b=\psi_B(e)$, 
$\theta_w=\theta_e$ if $w=\psi_W(e)$, and $\mu_w=\frac{1}{2}\sum_{e\in\EE(C_\G)_w}\sin(2\theta_e)$.
Furthermore, the map $\omega\mapsto\varphi_\omega$ defines a bijection between equivalence classes of Kasteleyn orientations on $C_\G\subset\SI$ and inverse square roots in $H^1(\SI;S^1)$
of the holonomy on $\SI$.
\end{proposition}

\begin{proof}
Fix $f\in\C^B$ and $w\in W$. Following the notations of Figure~\ref{fig:psi}, let us write $e=\psi_W^{-1}(w)$, $e'=R(e)$, and $b_1=\psi_B(e)$, $b_2=\psi_B(\bar{e})$, $b_3=\psi_B(e')$ for the three vertices
of $C_\G$ adjacent to $w$. By definition, we have
\begin{align*}
\left(K^\omega(C_\G,y)\circ \exp\left(\textstyle{-\frac{i}{2}}\theta_B\right)\right)&(f)(w)=\omega(w,b_1)\exp\left(-{\textstyle\frac{i}{2}}\theta_e\right)\cos(\theta_e)f(b_1)+\cr
	&+\omega(w,b_2)\exp\left(-{\textstyle\frac{i}{2}}\theta_e\right)\sin(\theta_e)f(b_2)+\omega(w,b_3)\exp\left(-{\textstyle\frac{i}{2}}\theta_{e'}\right)f(b_3),
\end{align*}
and
\begin{align*}
\left(\exp\left(\textstyle{-\frac{i}{2}}\theta_W\right)\circ\mu_W\circ\bar{\partial}_C^{\varphi}\right)(f)(w)&=\exp\left(-{\textstyle\frac{i}{2}}\theta_e\right)\Big(\varphi(w,b_1)
	\cos(\theta_e)f(b_1)+\cr
	+&\varphi(w,b_2)\,i\,\sin(\theta_e)f(b_2)-\varphi(w,b_3)\exp(i\theta_e)f(b_3)\Big).
\end{align*}
Therefore, the diagram commutes if and only if the 1-cochain $\varphi\colon\EE(C_\G)\to S^1$ is given by
\[
\varphi(w,b_j)=
\begin{cases}
\omega(w,b_1)&\text{if $j=1$;}\\ 
-i\,\omega(w,b_2)&\text{if $j=2$;}\\
-\exp\left(\textstyle{-\frac{i}{2}}(\theta_e+\theta_{e'})\right)\omega(w,b_3)&\text{if $j=3$.}
\end{cases}
\]
Using the fact that $\omega$ is a Kasteleyn orientation and the fact that the cone angles of the singularities of $\SI$ are odd multiples of $2\pi$, one checks that $\varphi$ is a 1-cocycle.
Furthermore, if $\gamma$ is any closed curve in $C_\G$, $\varphi(\gamma)^{-2}$ is nothing but the holonomy $\mathit{Hol}(\gamma)$. Therefore, $\varphi$ is a discrete spin structure.
The last statement follows easily from the definitions.
\end{proof}

\subsection{Dirac versus Laplace operators}
\label{sub:D-L}

In this whole paragraph, we shall assume that $\G$ is a graph isoradially embedded in a flat surface $\SI$ with trivial holonomy. On such a surface, one can fix a constant vector field $X$,
which we will be using to evaluate the associated discrete Dirac operators (recall Definition~\ref{def:Dirac}).

The following relation is well-known in the planar~\cite{Ken} and toric~\cite{BdT1} cases. The proof is straightforward.

\begin{proposition}
\label{prop:D-L1}
Given any graph $\G$ isoradially embedded in a flat surface with trivial holonomy and any 1-cocycle $\varphi_D\colon\EE(D_\G)\to\C^*$,
\[
-\,\partial_D^{\varphi_D}\circ\bar\partial^{\varphi_D}_D=\Delta_\G^\varphi\oplus\Delta_{\G^*}^{\varphi_*}
\]
as endomorphisms of $\C^{B(D)}=\C^{V(\G)}\oplus\C^{V(\G^*)}$.\qed
\end{proposition}

Note that the composition $-\bar\partial_D^{\varphi_D}\circ\partial^{\varphi_D}_D$ does not seem to be related to the discrete Laplace operator in any sensible way.
Therefore, it is not clear how to make sense of the square of the discrete Dirac operator
$(\Di^\varphi_D)^2=\left(\begin{smallmatrix}0&-\partial_D^\varphi\cr \bar\partial_D^\varphi & 0\end{smallmatrix}\right)^2
=\left(\begin{smallmatrix}-\partial_D^\varphi\bar\partial^\varphi_D&0\cr 0&-\bar\partial_D^\varphi\partial_D^\varphi\end{smallmatrix}\right)$.

On the other hand, considering the bipartite graph $C_\G$ instead of $D_\G$ leads to a more symmetrical result, as we now explain.
First, note that in the smooth setting, a direct computation leads to the equality
$\Di^2=\left(\begin{smallmatrix}-\partial\bar\partial&0\cr 0&-\bar\partial\partial\end{smallmatrix}\right)=\left(\begin{smallmatrix}\Delta-iA&0\cr 0&\Delta+iA\end{smallmatrix}\right)$,
where $A=\frac{\partial^2}{\partial x\partial y}-\frac{\partial^2}{\partial y\partial x}$ vanishes. As we will see, this is reflected in the discretization by $C_\G$.

Given an isoradially embedded graph $\G$, let $M$ denote the (isoradial) graph dual to $D_\G$, as illustrated below.
Also, let $\varepsilon$ denote the orientation of the edges of $M$ illustrated in the same figure.

\begin{figure}[h]
\labellist\small\hair 2.5pt
\pinlabel {$\G$} at 35 265
\pinlabel {$M=D_\G^*$} at 640 265
\endlabellist
\begin{center}
\includegraphics[height=2cm]{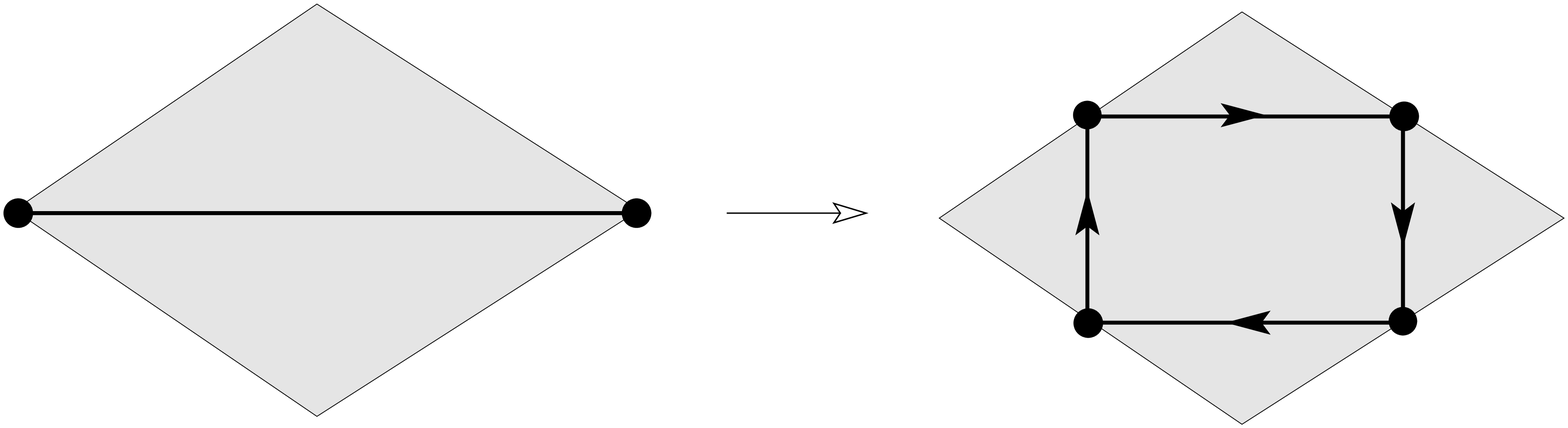}
\end{center}
\end{figure}

\noindent Note that $C_\G$ can be obtained from $M$ by replacing each vertex by a ``small'' edge between two vertices. (This operation is very natural in the setting of discrete complex analysis,
see~\cite[Section 2.2]{Cim1}.) 
Without loss of generality, we shall consider 1-cochains $\varphi$ on $C_\G$ that are trivial on the small edges of $C_\G$, and use the same notation $\varphi$ for the restriction to the edges of $M$. 
Finally, we shall denote by $A^\varphi_\varepsilon$ the associated normalized $\varphi$-twisted skew-adjacency operator on $\C^{V(M)}$.
In other ``words'',
\[
(A^\varphi_\varepsilon f)(v)=\frac{1}{\mu_v}\sum_{e=(v,v')}\varepsilon(e)\varphi(e)f(v')
\]
for $f\in\C^{V(M)}$ and $v\in V(M)$, the sum being over all edges of $M$ of the form $(v,v')$,
with $\varepsilon(e)=1$ if $\varepsilon$ orients $e$ from $v$ to $v'$, and $\varepsilon(e)=-1$ else.

\begin{proposition}
\label{prop:D-L2}
Given any graph $\G$ isoradially embedded in a flat surface with trivial holonomy and any 1-cocycle $\varphi\colon\EE(C_\G)\to\C^*$ which is trivial on the small edges of $C_\G$,
\[
(\Di_C^\varphi)^2
=\begin{pmatrix}-\partial_C^\varphi\bar\partial^\varphi_C&0\cr 0&-\bar\partial_C^\varphi\partial_C^\varphi\end{pmatrix}
=\frac{1}{2}\begin{pmatrix}\Delta^\varphi_M-iA^\varphi_\varepsilon&0\cr 0&\Delta^\varphi_M+iA^\varphi_\varepsilon\end{pmatrix}
\]
as endomorphisms of $\C^{V(C)}=\C^{V(M)}\oplus\C^{V(M)}$. In particular,
\[
-\left(\partial_C^\varphi\bar\partial^\varphi_C+\bar\partial_C^\varphi\partial^\varphi_C\right)=\Delta^\varphi_M.
\]
\end{proposition}
\begin{proof}
By definition, for any $b\in B:=B(C_\G)$ and $f\in\C^B$,
\[
-(\partial_C^\varphi\circ\bar\partial^\varphi_C)(f)(b)=\frac{-1}{2\mu_b}\sum_{e=(b,w)}\sum_{e'=(w,b')}\varphi(e)\varphi(e')\frac{\sin(\theta_e)\sin(\theta_{e'})}{\sin(\theta_w)\cos(\theta_w)}
\exp(i(\vartheta(e')-\vartheta(e)))f(b'),
\]
where $\theta_w$ stands for $\theta_{\psi^{-1}_W(w)}$ and $\vartheta$ for $\vartheta_X$. Let us denote the vertices of $C_\G$ around $b$ as illustrated below.

\begin{figure}[h]
\labellist\small\hair 2.5pt
\pinlabel {$b$} at 136 248
\pinlabel {$w$} at 130 174
\pinlabel {$b_1$} at 33 368
\pinlabel {$w_1$} at 120 322
\pinlabel {$b_2$} at 368 255
\pinlabel {$w_2$} at 280 307
\pinlabel {$b'$} at 263 375
\pinlabel {$\theta$} at 390 421
\pinlabel {$\theta'$} at 406 60
\pinlabel {$b_3$} at 53 150
\pinlabel {$b_4$} at 293 87
\endlabellist
\begin{center}
\includegraphics[height=4.5cm]{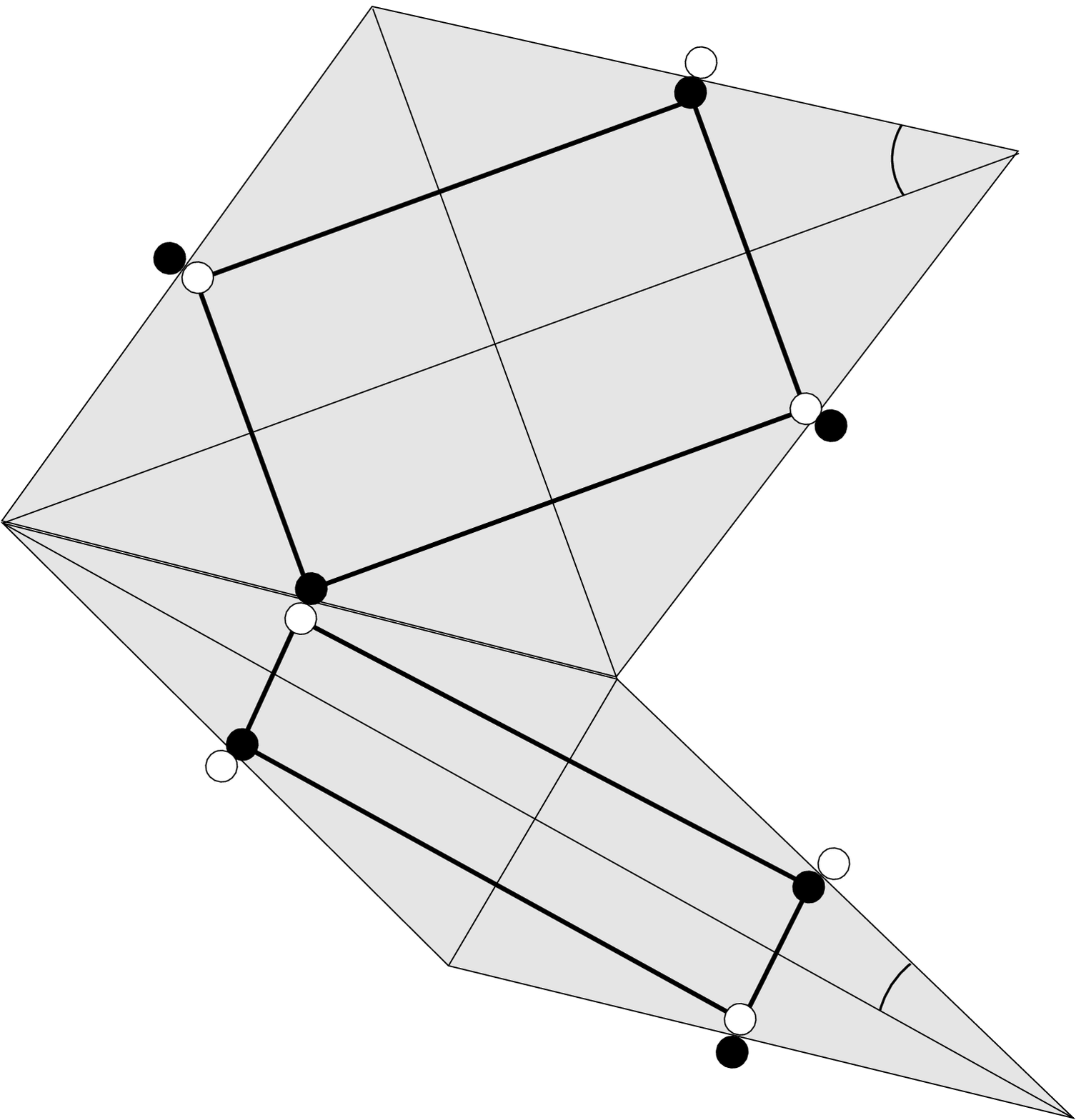}
\end{center}
\end{figure}

\noindent Also, we shall denote by $v$ (resp. $v_1,v_2$) the vertex of $M$ corresponding to $b$ and $w$ (resp. $b_1$ and $w_1$, $b_2$ and $w_2$).
Since $\varphi$ is a cochain and $\vartheta(\bar{e})-\vartheta(e)=\pm\pi$, the coefficient of $f(b)$ in the formula above is given by
\[
\frac{1}{2}\left(\tan(\theta)+\tan(\theta^*)+\tan(\theta')+\tan((\theta')^*)\right),
\]
where $\theta^*=\frac{\pi}{2}-\theta$ as usual. This coincides with the coefficient of $f(v)$ in
$\frac{1}{2}\Delta^\varphi_M(f)(v)$. As for the coefficient of $f(b')$, since $\varphi$ is a cocycle and $\vartheta(w_2,b')-\vartheta(b,w_2)=\frac{\pi}{2}$ while
$\vartheta(w_1,b')-\vartheta(b,w_1)=-\frac{\pi}{2}$, it vanishes as expected. Since $\varphi$ is trivial on the small edges and $\vartheta(w_1,b_1)-\vartheta(b,w_1)=\theta$,
the coefficient of $f(b_1)$ is equal to
\[
\frac{-1}{2\mu_b}\varphi(b,w_1)\frac{\cos(\theta)}{\cos(\theta)\sin(\theta)}\exp(i\theta)=\frac{-1}{2\mu_b}\varphi(v,v_1)(\tan(\theta^*)+i),
\]
which coincides with the coefficient of $f(v_1)$ in $\frac{1}{2}(\Delta^\varphi_M-iA^\varphi_\varepsilon)(f)(v)$.
Similarly, the coefficient of $f(b_2)$ is equal to
\[
\frac{-1}{2\mu_b}\varphi(b,w_2)\frac{\sin(\theta)}{\cos(\theta)\sin(\theta)}\exp(-i\theta^*)=\frac{-1}{2\mu_b}\varphi(v,v_2)(\tan(\theta)-i),
\]
which is the coefficient of $f(v_2)$ in $\frac{1}{2}(\Delta^\varphi_M-iA^\varphi_\varepsilon)(f)(v)$. The cases of $f(b_3)$ and $f(b_4)$ are treated similarly, leading to the equation
\[
-\partial_C^\varphi\bar\partial^\varphi_C=\frac{1}{2}(\Delta^\varphi_M-iA^\varphi_\varepsilon).
\]
The formula for $-\bar\partial_C^\varphi\partial^\varphi_C$ follows, since the coefficients of $\bar\partial_C$ and $\partial_C$ are complex conjugate.
\end{proof}

\subsection{Comparing discrete Dirac operators}
\label{sub:D-D}

Surprisingly, the discrete Dirac operators $\Di^\varphi_D$ on $D=D_\G$ and $\Di^\varphi_C$ on $C=C_\G$ are also related in a fairly natural way.
To present this result, it is convenient to adopt the following terminology:
we will say that vertices $v\in V(D)=\Lambda\cup\diamondsuit$ and $v'\in V(C)=B\cup W$ are {\em adjacent\/}, denoted $v\sim v'$, if they are linked by an edge in the right-hand side of the
illustration below.

\begin{figure}[h]
\labellist\small\hair 2.5pt
\pinlabel {$V(\G)$} at 35 265
\pinlabel {$V(D)\cup V(C)$} at 610 265
\endlabellist
\begin{center}
\includegraphics[height=2.5cm]{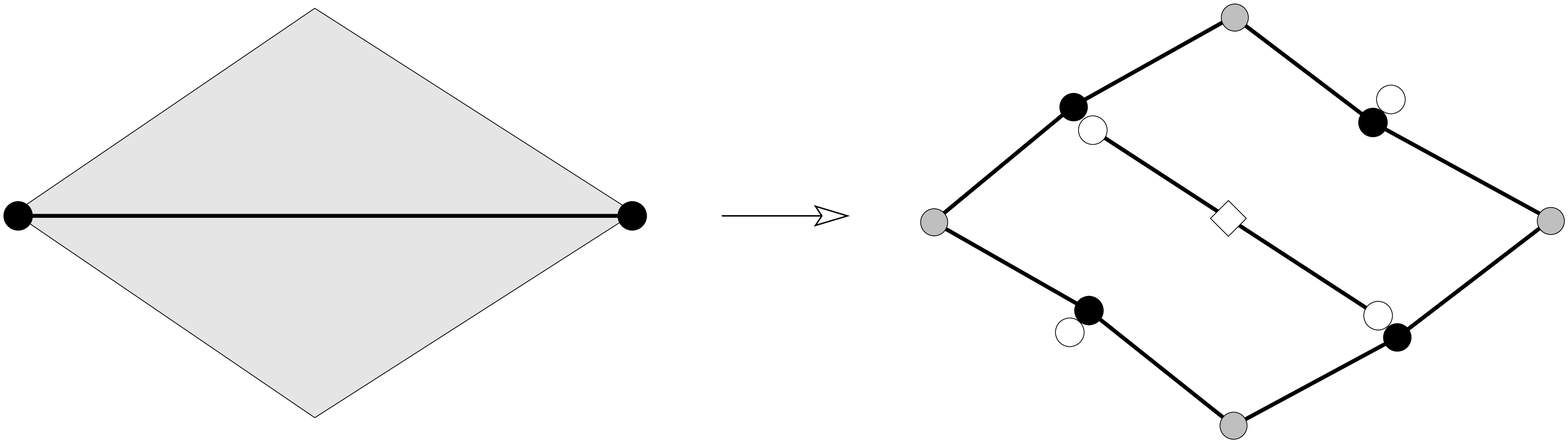}
\end{center}
\end{figure}

\noindent Clearly, any 1-cocycle $\varphi$ on the graph above naturally induces 1-cocycles $\varphi_D$ on $D$ and $\varphi_C$ on $C$. We shall denote by
$h^\varphi_{DC}\colon\C^{V(D)}\to\C^{V(C)}$ and $h^\varphi_{CD}\colon\C^{V(C)}\to\C^{V(D)}$ the associated twisted adjacency operators, i.e. the operators defined by
\[
(h_{DC}^\varphi f)(v')=\frac{\deg(v')}{2}\sum_{v'\sim v\in V(D)}\varphi(v',v)f(v)\quad\text{and}\quad
(h_{CD}^\varphi g)(v)=\hspace{-.2cm}\sum_{v\sim v'\in V(C)}\hspace{-.2cm}\varphi(v,v')g(v')
\]
for $f\in\C^{V(D)}$, $g\in\C^{V(C)}$, $v\in V(D)$ and $v'\in V(C)$, where $\deg(v')\in\{1,2\}$ denotes the number of vertices in $V(D)$ adjacent to $v'$.

\begin{proposition}
\label{prop:d-d}
Given any graph $\G$ isoradially embedded in a flat surface with trivial holonomy and any 1-cocycle $\varphi$ as above, the following diagram commutes:
\[
\xymatrix{
\C^{V(D)}\ar[d]_{h^\varphi_{D C}}\ar[rr]^{\Di^{\varphi_D}_D}& \hspace{1cm} & \C^{V(D)}\\
\C^{V(C)}\ar[rr]^{\Di^{\varphi_C}_C} & \hspace{1cm} & \C^{V(C)}\ar[u]_{\mu_D^{-1}\circ h_{CD}^\varphi\circ\mu_C},}
\]
where $\mu_D=(\mu_v)_{v\in V(D)}$ and $\mu_C=(\mu_{v'})_{v'\in V(C)}$ are the usual diagonal operators.
\end{proposition}
\begin{proof}
Let $f$ be an element of $\C^{V(D)}$. Fixing $z\in\diamondsuit\subset V(D)$, let us denote by $(v_1,v_2,v_3,v_4)$ the vertices of $\Lambda$ around $z$ numbered counterclockwise, and by $\theta$
the corresponding half-rhombus angle (assuming $v_1$ and $v_3$ are vertices of the primal graph $\G$). By definition, we have
\begin{align*}
\left(h^\varphi_{CD}\circ\mu_C\circ\Di^{\varphi_C}_C\circ h^\varphi_{D C}\right)(f)(z)&=
		\sum_{z\sim w\in W}\sum_{(w,b)\in E(C)}\sum_{b\sim v\in\Lambda}\varphi_D(z,v)e^{i\vartheta_X(w,b)}\sin(\theta_{wb})f(v)\cr
	&=e^{i\vartheta_X(v_3,v_1)}\sum_j\varphi_D(z,v_j)\,c_j\,f(v_j),
\end{align*}
where the coefficients $c_j$ are given by $c_1=i\cos(\theta)+e^{-i\theta^*}=\sin(\theta)$, $c_2=\sin(\theta)+ie^{i\theta}=i\cos(\theta)$, $c_3=-i\cos(\theta)+ie^{i\theta}=-\sin(\theta)$ and
$c_4=-\sin(\theta)+e^{-i\theta^*}=-i\cos(\theta)$. Therefore,
\[
\left(h^\varphi_{CD}\circ\mu_C\circ\Di^{\varphi_C}_C\circ h^\varphi_{D C}\right)(f)(z)=\left(\mu_D\circ\bar\partial^{\varphi_D}_D\right)(f)(z)=\left(\mu_D\circ\Di^{\varphi_D}_D\right)(f)(z).
\]
Furthermore, fixing $v\in\Lambda$, we have
\begin{align*}
\left(h^\varphi_{CD}\circ\mu_C\circ\Di^{\varphi_C}_C\circ h^\varphi_{D C}\right)&(f)(v)=
		-\frac{1}{2}\sum_{v\sim b\in B}\sum_{(b,w)\in E(C)}\sum_{w\sim z\in\diamondsuit}\varphi_D(v,z)e^{-i\vartheta_X(b,w)}\sin(\theta_{bw})f(z)\cr
	&=-\frac{1}{2}\sum_{(v,z)\in E(D)}\varphi_D(v,z)e^{-i\vartheta_X(v,z)}\left(\sin(\theta_z)+i\cos(\theta_z)+e^{-i\theta_z^*}\right)f(z)\cr
	&=-\sum_{(v,z)\in E(D)}\varphi_D(v,z)e^{-i\vartheta_X(v,z)}\sin(\theta_z)f(z)\cr
	&=-\left(\mu_D\circ\bar\partial^{\varphi_D}_D\right)(f)(v)\cr
	&=\left(\mu_D\circ\Di^{\varphi_D}_D\right)(f)(v),
\end{align*}
and the proof is completed.
\end{proof}

\bigbreak

Let us conclude this section with one last summarizing remark. Given any isoradially embedded graph $\G\subset\SI$ with critical weights $x$, Corollary~\ref{cor:det} and Proposition~\ref{prop:Kast-D} give
the relation between determinants
\[
\left|\KW^\varphi(\G,x)\right|^2\stackrel{\cdot}{=}\left|K^\varphi(C_\G,y(x))\right|^2\stackrel{\cdot}{=}\left|\bar\partial^\varphi_C\right|^2,
\]
where $\stackrel{\cdot}{=}$ stands for the equality between functions of $\varphi\in H^1(\SI;\C^*)\simeq(\C^*)^{2g}$ up to a multiplicative constant.
Furthermore, if $\varphi$ is unitary, i.e. belongs to $(S^1)^{2g}$, then $\bar\partial^\varphi_D$ and $\partial^\varphi_D$ are adjoint matrices, which are square if $g=1$.
Assuming that $\SI$ is a torus, Proposition~\ref{prop:D-L1} therefore leads to
\[
\left|\bar\partial^\varphi_D\right|^2=\left|\partial^\varphi_D\circ\bar\partial^\varphi_D\right|\stackrel{\cdot}{=}|\Delta_\G^\varphi\oplus\Delta_{\G^*}^\varphi|\stackrel{\cdot}{=}
|\Delta_\G^\varphi|^2.
\]
(The last equality is well-known, see e.g.~\cite{BdT1}.) Finally, Proposition~\ref{prop:d-d} seems to hint at a relation of the form $|\bar\partial^\varphi_C|\stackrel{\cdot}{=}|\bar\partial^\varphi_D|$,
which would imply the equality $\left|\KW^\varphi(\G,x)\right|\stackrel{\cdot}{=}|\Delta_\G^\varphi|$. This equality actually holds under the conditions stated above, and only under these conditions,
as was proved in~\cite[Theorem 4.6]{Cim3}.

%%%%%%%%%%%%%%%%%%%%%%

\section{Generalized s-holomorphicity}
\label{sec:s-holo}

This section builds on the previous ones to obtain the main results of this article. In the first two paragraphs, we introduce a notion of s-holomorphicity valid for any weighted surface graph,
generalizing the classical definition of Chelkak-Smirnov~\cite{CS09} which corresponds to the planar isoradial case. Moreover, we give three alternative viewpoints on this notion, each
involving one of the three operators (Kac-Ward, Kasteleyn, Dirac) studied in Section~\ref{sec:rel} (Theorem~\ref{thm:s} and Corollary~\ref{cor:s}). In the final two subsections, we show that several
crucial properties of s-holomorphic functions on isoradial graphs extend to our setting. First, the minors of the Kac-Ward matrices are nothing but generalized spin-observables, which are
automatically s-holomorphic (subsection~\ref{sub:inv}). Also, it is possible to define a discrete version of the integral of the square of an s-holomorphic function (subsection~\ref{sub:F2}).

\subsection{The kernel of the Kac-Ward operator}
\label{sub:ker}

As above, let $(\G,x)\subset\SI$ be an arbitrary weighted surface graph. In this paragraph, we will analyse the kernel of the associated Kac-Ward operator $\KW=\KW(\G,x)$,
where the 1-cochain $\varphi$ is taken to be trivial. (The discussion below extends to any real-valued cochain, but the statements become unnecessarily cumbersome.)

First note that the corresponding Kasteleyn operator $K^\omega=K^\omega(C_\G,y)$ restricts to a real operator $\mathrm{K}^\omega(C_\G,y)\colon\R^B\to\R^W$.
The commutative diagram of Theorem~\ref{thm:corr} can therefore be completed into
\[
\xymatrix{
\C^\EE\ar[rr]^{\KW(\G,x)} & \hspace{1cm} & \C^\EE\\
\C^B\ar[u]^{(I-qR)\circ D^{-1/2}\circ\psi_B}_\simeq\ar[rr]^{K^\omega(C_\G,y)} & \hspace{1cm} & \C^W\ar[u]_{(I-ixJ)\circ D^{-1/2}\circ\psi_W}^\simeq \\
\R^B\ar@{^{(}->}[u]\ar[rr]^{\mathrm{K}^\omega(C_\G,y)} & \hspace{1cm} & \R^W,\ar@{^{(}->}[u]}
\]
where the choice of the square root of $D$ is determined by the Kasteleyn orientation $\omega$.

For $e\in\EE$, let us write
\[
\ell(e)=D_e^{-1/2}\cdot\R=\exp\left(-{\textstyle\frac{i}{2}}a_e\right)\cdot\R,
\]
which does not depend on the square root of $D$. By definition, both $(D^{-1/2}\circ\psi_B)(\R^B)$ and $(D^{-1/2}\circ\psi_W)(\R^W)$ coincide with the real vector subspace of $\C^\EE$ given by 
\[
\L=\{f\in\C^\EE\;|\;f(e)\in\ell(e)\text{ for all $e\in\EE$}\}.
\]
Since $\ell(\bar{e})=i\ell(e)$, $I-ixJ$ leaves $\L$ invariant. Since $\ell(R(e))=q_e^{-1}\ell(e)$, the same holds for $I-qR$.
Therefore, we obtain the following commutative diagram:
\[
\xymatrix{
\L\ar[rr]^{\KW(\G,x)} & \hspace{1cm} & \L\\
\R^B\ar[u]^{(I-qR)\circ D^{-1/2}\circ\psi_B}_\simeq\ar[rr]^{\mathrm{K}^\omega(C_\G,y)} & \hspace{1cm} & \R^W.\ar[u]_{(I-ixJ)\circ D^{-1/2}\circ\psi_W}^\simeq}
\]

Given any real line $\ell\subset\C$, we shall denote by $\Pr(-;\ell)\colon\C\to\ell$ the orthogonal projection onto $\ell$. We shall also simply write $\Pr(-;u)$ for the orthogonal projection onto
the real line $u\cdot\R$ generated by $u\in\C^*$. Following Lis~\cite{Lis}, let us consider the real linear map $S\colon\C^\diamondsuit\to\L$ defined by
\[
(SF)(e)=\sin(\theta_e/2)\Pr(F(z_e);\ell(e))
\]
for $F\in\C^\diamondsuit$ and $e\in\EE$. The following result generalizes~\cite[Proposition 2.1]{Lis}, which deals with the planar isoradial case.

\begin{proposition}
\label{prop:ker}
Fix an element $F$ of $\C^\diamondsuit$ and a vertex $v$ of $\G$. Then, $f=S(F)\in\L$ satisfies $(\KW f)(\bar{e})=0$ for all $e\in\EE_v$
if and only if
\[
\Pr\left(F(z_e)\,;\,i\exp\left(-\textstyle{\frac{i}{2}}(a_e+\theta_e)\right)\right)=\Pr\left(F(z_{e'})\,;\,i\exp\left(\textstyle{-\frac{i}{2}}(a_{e'}-\theta_{e'})\right)\right)\,
\exp\left(\textstyle{\frac{i}{2}}(\beta_e-\theta_e-\theta_{e'})\right)
\]
for all $e\in\EE_v$, where $e'$ stands for $R(e)$.
\end{proposition}
\begin{proof}
Using the notations of subsection~\ref{sub:corr}, the first statement means that $(J\circ \KW)(f)$ vanishes on $\EE_v$. Consider the endomorphism of $\C^\EE$ given by
$(I-qR)x^{-1}$. It is clearly an isomorphism (recall the proof of Corollary~\ref{cor:det}), which splits as a direct sum of automorphisms of $\C^{\EE_v}$ according to the
partition $\EE=\bigsqcup_{v}\EE_v$. Therefore, $(J\circ\KW)(f)$ vanishes on $\EE_v$ if and only if $((I-qR)x^{-1}J\;\KW)(f)$ vanishes on $\EE_v$. In other words,
for all $e\in\EE_v$,
\begin{align*}
0&=((I-qR)x^{-1}J\;\KW)(f)(e)\cr
&=x_e^{-1}(\KW f)(\bar{e})-q_e\,x^{-1}_{e'}(\KW f)(\bar{e'})\cr
	&=x_e^{-1}\Big(f(\bar{e})-x_e\sum_{e''\in \EE_v\setminus\{e\}}\exp\left({\textstyle\frac{i}{2}\alpha_\lambda(\bar{e},e'')}\right)f(e'')\Big)\cr
	&\phantom{=}-\exp\left({\textstyle\frac{i}{2}}\beta_e\right)x_{e'}^{-1}\Big(f(\bar{e'})-x_{e'}
	\sum_{e''\in \EE_v\setminus\{e'\}}\exp\left({\textstyle\frac{i}{2}\alpha_\lambda(\bar{e'},e'')}\right)f(e'')\Big)\cr
	&=x_e^{-1}f(\bar{e})+i\exp\left({\textstyle\frac{i}{2}}\beta_e\right)f(e')-\exp\left({\textstyle\frac{i}{2}}\beta_e\right)x_{e'}^{-1}f(\bar{e'})+if(e).
\end{align*}
Therefore, $f=S(F)$ satisfies $(\KW f)(\bar e)=0$ for all $e\in\EE_v$ if and only if
\[
x_e^{-1}f(\bar e)+if(e)=\exp\left(\textstyle{\frac{i}{2}}\beta_e\right)\left(x_{e'}^{-1}f(\overline{e'})-if(e')\right)
\]
for all $e\in\EE_v$. Using the parametrization $x_e=\tan(\theta_e/2)$ of the weights, the definition of $f=S(F)$, and the fact that $\ell(e)$ and $\ell(\bar{e})$ are orthogonal, the left-hand side reads
\begin{align*}
x_e^{-1}f(\bar e)+if(e)&=\sin(\theta_e/2)^{-1}\left(\cos(\theta_e/2)f(\bar{e})+i\sin(\theta_e/2)f(e)\right)\cr
	&=\cos(\theta_e/2)\Pr(F(z_e);\ell(\bar{e}))+i\sin(\theta_e/2)\Pr(F(z_e);\ell(e))\cr
	&=\Pr((\cos(\theta_e/2)+i\sin(\theta_e/2))F(z_e);\ell(\bar{e}))\cr
	&=\exp\left(\textstyle{\frac{i}{2}}\theta_e\right)\Pr\left(F(z_e);\exp\left(-\textstyle{\frac{i}{2}}\theta_e\right)\ell(\bar{e})\right)\cr
	&=\exp\left(\textstyle{\frac{i}{2}}\theta_e\right)\Pr\left(F(z_e);i\exp\left(-\textstyle{\frac{i}{2}}(a_e+\theta_e)\right)\right).
\end{align*}
Similar considerations for the right-hand side lead to
\[
x_{e'}^{-1}f(\overline{e'})-if(e')=\exp\left(-\textstyle{\frac{i}{2}}\theta_{e'}\right)\Pr\left(F(z_{e'});i\,\exp\left(\textstyle{-\frac{i}{2}}(a_{e'}-\theta_{e'})\right)\right),
\]
and the proposition follows.
\end{proof}

Note that the left-hand side of the equality in Proposition~\ref{prop:ker} can be written as
\[
\Pr\left(F(z_e);i\exp\left(-\textstyle{\frac{i}{2}}(a_e+\theta_e)\right)\right)=
-i\,\exp\left(-\textstyle{\frac{i}{2}}\theta_e\right)D_e^{-1/2}\mathit{Re}\left(i\,D_e^{1/2}\exp\left(\textstyle{\frac{i}{2}}\theta_e\right)F(z_e)\right),
\]
while
\[
\Pr\left(F(z_{e'});i\exp\left(\textstyle{-\frac{i}{2}}(a_{e'}-\theta_{e'})\right)\right)=
-i\,\exp\left(\textstyle{\frac{i}{2}}\theta_{e'}\right)D_{e'}^{-1/2}\mathit{Re}\left(i\,D_{e'}^{1/2}\exp\left(-\textstyle{\frac{i}{2}}\theta_{e'}\right)F(z_{e'})\right).
\]
Therefore, this equality is equivalent to
\[
\mathit{Re}(i\,D_e^{1/2}\exp\left(\textstyle{\frac{i}{2}}\theta_e\right)F(z_e))=\varepsilon(e)\,\mathit{Re}(i\,D_{e'}^{1/2}\exp\left(-\textstyle{\frac{i}{2}}\theta_{e'}\right)F(z_{e'})),\tag{$\star$}
\]
where $\varepsilon(e)=q_e D_{e'}^{-1/2}D_e^{1/2}=\pm 1$ depends on the Kasteleyn orientation $\omega$.

\subsection{Three viewpoints on s-holomorphicity}
\label{sub:3}

Assume that $\G$ is isoradially embedded in the plane with critical weights $x_e=\tan(\theta_e/2)$.
If $\lambda$ is chosen to be a constant vector field, then $\beta_e-\theta_e-\theta_{e'}$ vanishes for all $e\in\EE$. Therefore the equality in Proposition~\ref{prop:ker} simply reads
\[
\Pr(F(z_e);\ell_{e,e'})=\Pr(F(z_{e'});\ell_{e,e'}),
\]
where $\ell_{e,e'}=i\exp\left(-\textstyle{\frac{i}{2}}(a_e+\theta_e)\right)\cdot\R=i\exp\left(-\textstyle{\frac{i}{2}}(a_{e'}-\theta_{e'})\right)\cdot\R$.
This equality defines the fact that $\exp\left(i\textstyle{\frac{\pi}{4}}\right)F$ is ``s-holomorphic''~\cite{CS09}. Hence, the equality in Proposition~\ref{prop:ker}, or equivalently, the fact of lying in the kernel of
the Kac-Ward operator, should be understood as a generalized s-holomorphicity condition valid for any weighted surface graph. This motivates the following terminology.

\begin{definition}
\label{def:s}
Let $(\G,x)\subset\SI$ be an arbitrary surface graph with weights parametrized by $x_e=\tan(\theta_e/2)$. A function $F\in\C^\diamondsuit$ is called {\em s-holomorphic\/} around $v\in V(\G)$ if
\[
\Pr\left(F(z_e);\left[i\exp(i(a_e+\theta_e))\right]^{-\frac{1}{2}}\right)=
\Pr\left(F(z_{e'});\left[i\exp(i(a_{e'}-\theta_{e'}))\right]^{-\frac{1}{2}}\right)\exp\left(\textstyle{\frac{i}{2}(\beta_e-\theta_e-\theta_{e'}})\right)
\]
for all $e\in\EE_v$, where $e'$ stands for $R(e)$.
\end {definition}

Note that the choice of the square root is irrelevant.

We shall now build on the results of Section~\ref{sec:rel} to give three viewpoints on the generalized notion of s-holomorphicity defined above.
Recall that a Kasteleyn orientation $\omega$ on $C_\G\subset\SI$ determines a square root $D^{1/2}=(D_e^{1/2})_{e\in\EE}$ of $D=(D_e)_{e\in\EE}=(\exp(ia_e))_{e\in\EE}$
as explained in Theorem~\ref{thm:corr}, and that these square roots satisfy $D_{R(e)}^{1/2}=\varepsilon(e)q_eD_e^{1/2}$ for all $e\in\EE$, with $\varepsilon(e)=\pm 1$. This allows us to define
$T\colon\C^\diamondsuit\to\R^B$ by
\[
(TF)(b)=\sum_{e\in\EE_v}\mathit{Re}\left(\varepsilon_b(e)D_e^{1/2}\sin(\theta_e/2)F(z_e)\right),
\]
where $v\in V(\G)$ denotes the vertex of $\G$ closest to $b\in B$, i.e. $v=o(e_0)$ if $b=\psi_B(e_0)$, and $\varepsilon_b(e)=\varepsilon(e_0)\varepsilon(R(e_0))\cdots\varepsilon(R^{k-1}(e_0))$
if $e=R^k(e_0)$, $0\le k<\deg(v)$. Finally, set $T':=\exp\left(\textstyle{\frac{i}{2}}\theta_B\right)\circ T$,
where $\theta_B=(\theta_{\psi^{-1}_B(b)})_{b\in B}$.

\begin{theorem}
\label{thm:s}
The maps $S\colon\C^\diamondsuit\to\L$, $T\colon\C^\diamondsuit\to\R^B$ and $T'\colon\C^\diamondsuit\to T'(\C^\diamondsuit)\subset\C^B$ are $\R$-linear isomorphisms such that,
for any $F\in\C^\diamondsuit$, the following are equivalent:
\begin{romanlist}
\item{$\exp\left(i\textstyle{\frac{\pi}{4}}\right)F$ is s-holomorphic.}
\item{$S(F)\in\L$ lies in the kernel of the Kac-Ward operator $\KW(\G,x)$.}
\item{$T(F)\in\R^B$ lies in the kernel of the Kasteleyn operator $\mathrm{K}^\omega(C_\G,y)$.}
\end{romanlist}
Furthermore, if $\G$ is isoradially embedded in a flat surface $\SI$ with critical weights, then these three conditions are equivalent to:
\begin{romanlist}{\setcounter{enumi}{3}}
\item{$T'(F)\in\C^B$ lies in the kernel of the discrete $\overline{\partial}$-operator $\overline{\partial}^{\varphi_\omega}_C$.}
\end{romanlist}
\end{theorem}

\begin{proof}
The map $S\colon\C^\diamondsuit\to\L$ is clearly an isomorphism of real vector spaces, with inverse $(S^{-1}f)(z)=\sin(\theta_e/2)^{-1}(f(e)+f(\bar{e}))$ for $f\in\C^\EE$ and $z=z_e\in\diamondsuit$.
By definition, $\exp\left(i\textstyle{\frac{\pi}{4}}\right)F$ is s-holomorphic if and only if $F$ satisfies the equality in Proposition~\ref{prop:ker}; hence, $(i)$ and $(ii)$ are equivalent by this proposition.

Now, let us define $T_0\colon\C^\diamondsuit\to\R^B$ as the unique $\R$-linear isomorphism such that $(I-qR)\circ D^{-1/2}\circ\psi_B\circ T_0=S$.
By Theorem~\ref{thm:corr} and the discussion in subsection~\ref{sub:ker}, we have the commutative diagram
\[
\xymatrix{
\C^\diamondsuit\ar[r]^S_\simeq \ar@/_{10pt}/[dr]_{T_0} & \L\ar[rr]^{\KW(\G,x)} & \hspace{1cm} & \L'\\
& \R^B\ar[u]^\simeq \ar[rr]^{\mathrm{K}^\omega(C_\G,y)} & \hspace{1cm} & \R^W.\ar[u]_\simeq}
\]
This implies that $\KW(SF)=0$ if and only if $T_0(F)$ lies in the kernel of $\mathrm{K}^\omega(C_\G,y)$.
Therefore, it remains to check that $T_0$ coincides with the map $T$ defined before the statement of the theorem (up to a real scalar factor).
To do so, first observe that the automorphism $I-qR$ of $\C^\EE$ splits as a direct sum $\bigoplus_{v\in V(\G)}(I-qR)_v$ according to the partition $\EE=\bigsqcup_{v\in V(\G)}\EE_v$, and that
$2(I-qR)_v^{-1}=I+qR+\dots+(qR)^{\deg(v)-1}$. This leads to the following computation, where $F\in\C^\diamondsuit$, $b=\psi_B(e)\in B$ and $v=o(e)\in V(\G)$:
\begin{align*}
2(T_0F)(b)&=2(\psi_B^{-1}\circ D^{1/2}\circ(I-qR)^{-1}\circ S)(F)(b)\cr
	&=2D_e^{1/2}(I-qR)_v^{-1}(SF)(e)\cr
	&=D_e^{1/2}\sum_{k=0}^{\deg(v)-1}q_e\cdots q_{R^{k-1}(e)}(SF)(R^k(e))\cr
	&=\sum_{e'\in\EE}\varepsilon_b(e')D_{e'}^{1/2}(SF)(e')\cr
	&=\sum_{e'\in\EE}\varepsilon_b(e')D_{e'}^{1/2}\sin(\theta_{e'}/2)\Pr(F(z_{e'});\ell(e'))\cr
	&=\sum_{e'\in\EE_v}\mathit{Re}\left(\varepsilon_b(e')D_{e'}^{1/2}\sin(\theta_{e'}/2)F(z_{e'})\right).
\end{align*}
This shows that the map $T=2T_0$ is an isomorphism, and that conditions $(ii)$ and $(iii)$ are equivalent. The map $T'=\exp\left(\textstyle{\frac{i}{2}}\theta_B\right)\circ T$
being the composition of two isomorphisms, it is one as well.

Finally, let us assume that $\G$ is isoradially embedded in a flat surface $\SI$ with weights $x_e=\tan(\theta_e/2)$. By Proposition~\ref{prop:Kast-D}, we have the commutative diagram
\[
\xymatrix{
\C^\diamondsuit\ar[r]^T_\simeq \ar@/_{10pt}/[dr]_{T'} & \R^B\ar[rr]^{\mathrm{K}^\omega(C_\G,y)} & \hspace{1cm} & \R^W\\
& \L_B\ar[u]^\simeq\ar[rr]^{\overline{\partial}_C^{\varphi_\omega}} & \hspace{1cm} & \L_W,\ar[u]_\simeq},
\]
where $\L_B=\exp\left(\textstyle{\frac{i}{2}}\theta_B\right)(\R^B)$ and $\L_W=\exp\left(\textstyle{\frac{i}{2}}\theta_W\right)(\R^W)$. Therefore, $T(F)$ lies in the kernel of $\mathrm{K}^\omega(C_\G,y)$
if and only if $T'(F)$ lies in the kernel of $\overline{\partial}_C^{\varphi_\omega}$.
\end{proof}

This theorem admits the following corollary.

\begin{corollary}
\label{cor:s}
Let $s\mathcal{O}$ denote the real vector subspace of $\C^\diamondsuit$ consisting of functions $F$ such that $\exp\left(i\textstyle{\frac{\pi}{4}}\right)F$ is s-holomorphic.
\begin{romanlist}
\item{The map $S\colon\C^\diamondsuit\to\C^\EE$ restricts to an $\R$-linear isomorphism $s\mathcal{O}\simeq\L\cap\mathrm{ker}(\KW)$.}
\item{The map $\widetilde{T}\colon\C^\diamondsuit\to\R^B$ defined by
\[
(\widetilde{T}F)(b)=\mathit{Re}\left(i\,D_e^{1/2}\,\exp\left(\textstyle{-\frac{i}{2}}\theta_e\right)F(z_e)\right)
\]
for $F\in\C^\diamondsuit$ and $b=\psi_B(e)\in B$ restricts to an isomorphism $s\mathcal{O}\simeq\mathrm{ker}(\mathrm{K}^\omega)$.}
\item{If $\G$ is isoradially embedded in a flat surface $\SI$ with critical weights, then the map $\widetilde{T'}\colon\C^\diamondsuit\to\C^B$ defined by
\[
(\widetilde{T'}F)(b)=i\,D_e^{1/2}\,\Pr\left(F(z_e);i\exp\left(\textstyle{-\frac{i}{2}}(a_e-\theta_e)\right)\right)
\]
for $F\in\C^\diamondsuit$ and $b=\psi_B(e)\in B$ restricts to an isomorphism $s\mathcal{O}\simeq\mathrm{ker}(\overline{\partial}^{\varphi_\omega}_C)\cap \widetilde{T'}(\C^\diamondsuit)$.}
\end{romanlist}
\end{corollary}

Note that in the planar isoradial case, $\widetilde{T}(F)$ is the so-called {\em real spinor\/} associated to $F$ by Chelkak-Smirnov in~\cite[Lemma 3.4]{CS09}. To be more precise,
the vertices of $C_\G$ play the role of the vertices of the double cover $\widehat{\Upsilon}$ in~\cite{CS09}, the choice of this double cover, or
{\em spin structure\/}~\cite[Definition 9]{Mer}, corresponding to the choice of a Kasteleyn orientation on $C_\G$ (see~\cite{C-RI}).
Furthermore, the fact that $\widetilde{T}(F)$ lies in the kernel of the Kasteleyn operator translates into the {\em propagation equation\/}~\cite[Equation $(3.6)$]{CS09}, or
the {\em Dirac equation\/}~\cite[Equation $(4.5)$]{Mer} (see also~\cite[Section 4.2]{Dub}). Therefore, the statement $(ii)$ above should be understood as a generalized propagation equation, valid for any weighted surface graph.

\begin{proof}[Proof of Corollary~\ref{cor:s}]
The first point follows immediately from Theorem~\ref{thm:s}. For the second point, it remains to check that the restriction of $T$ to $s\mathcal{O}$ coincides with the map $\widetilde{T}$ (up
to a real scalar factor). This follows from the computation below, where we make use of Equation $(\star)$ in the fourth equality, and of the notation $e'=R(e)$:
\begin{align*}
((I-qR)\circ&D^{-1/2}\circ\psi_B\circ\widetilde{T})(F)(e)=(D^{-1/2}\circ\psi_B\circ\widetilde{T})(F)(e)-q_e(D^{-1/2}\circ\psi_B\circ\widetilde{T})(F)(e')\cr
	&=D_e^{-1/2}(\widetilde{T}F)(\psi_B(e))-q_e\,D_{e'}^{-1/2}(\widetilde{T}F)(\psi_B(e'))\cr
	&=D_e^{-1/2}\mathit{Re}\left(i\;D_e^{1/2}\,\exp\left(\textstyle{-\frac{i}{2}}\theta_e\right)F(z_e)\right)-q_e\,D_{e'}^{-1/2}\mathit{Re}\left(i\;D_{e'}^{1/2}\,\exp\left(\textstyle{-\frac{i}{2}}\theta_{e'}\right)F(z_{e'})\right)\cr
	&=D_e^{-1/2}\mathit{Re}\left(i\;D_e^{1/2}\,\exp\left(\textstyle{-\frac{i}{2}}\theta_e\right)F(z_e)\right)-D_e^{-1/2}\mathit{Re}\left(i\;D_e^{1/2}\,\exp\left(\textstyle{\frac{i}{2}}\theta_e\right)F(z_e)\right)\cr
	&=D_e^{-1/2}\mathit{Re}\left(D_e^{1/2}\,2\sin(\theta_e/2)F(z_e)\right)\cr
	&=2\sin(\theta_e/2)\Pr(F(z_e);D_e^{-1/2}\cdot\R)\cr
	&=2(SF)(e).
\end{align*}
As for the third point, it follows from the second one and from the computation below:
\begin{align*}
(\widetilde{T'}F)(b)&=\left(\exp\left(\textstyle{\frac{i}{2}}\theta_B\right)\circ\widetilde{T}\right)(F)(b)\cr
	&=\exp\left(\textstyle{\frac{i}{2}}\theta_e\right)\mathit{Re}\left(i\;D_e^{1/2}\,\exp\left(\textstyle{-\frac{i}{2}}\theta_e\right)F(z_e)\right)\cr
	&=i\,D_e^{1/2}\Pr\left(F(z_e);i\exp\left(\textstyle{-\frac{i}{2}}(a_e-\theta_e)\right)\right).
\end{align*}
This concludes the proof.
\end{proof}

Let us conclude this paragraph on s-holomorphicity with one final result. It is natural to wonder about the connection between (generalized) s-holomorphicity and continuous 
holomorphicity, especially on Riemann surfaces. Obviously, there is no such relation in full generality when the weights of the graph are not related to the
conformal structure of the underlying surface. In the critical isoradial case however, it is possible to use Corollary~\ref{cor:s} above together with Theorems 2.5 and 3.5
of~\cite{Cim1} to obtain convergence results. The precise definition of all the objects involved being quite cumbersome, we refer the reader to~\cite{Cim1} for details.

\begin{corollary}
\label{cor:conv}
Let $\Sigma$ be a flat surface with cone-type singularities and trivial holonomy. Consider a sequence $\G_n$ of graphs $\delta_n$-isoradially embedded in $\Sigma$ with
$\lim_n\delta_n=0$, such that all rhombi angles of all these isoradial graphs belong to the interval $[\eta,\pi-\eta]$ for some $\eta>0$.
Finally, let $F_n\in\C^{\diamondsuit_n}$ be a sequence of functions such that $\exp\left(i\textstyle{\frac{\pi}{4}}\right)F_n$ is s-holomorphic
(with respect to a constant vector field).
\begin{romanlist}
\item Assume that $F_n\in\C^{\diamondsuit_n}$ converges to $F\colon\Sigma\to\C$ in the following sense: for any sequence $z_n\in\diamondsuit_n$ converging in $\Sigma$,
$F_n(z_n)$ converges to $F(\lim_n z_n)$ in $\C$. Then $F$ is holomorphic in $\Sigma$.
\item Assume that all cone angles are odd multiples of $2\pi$ and that the discrete spinors $\widetilde{T'}F_n\in\C^{B_n}$ converge to a section $\psi$ of a fixed
spin structure, interpreted as a line bundle $L\to\Sigma$. Then, $\psi$ is a holomorphic spinor.
\end{romanlist}
\end{corollary}
\begin{proof}
Since $\Sigma$ has trivial holonomy, it makes sense to talk about a constant vector field on $\Sigma$. With respect to such a vector field, the definition of s-holomorphicity
for isoradial graphs is the usual one. As in~\cite[Lemma 3.2]{CS09}, one then checks that if $F\in\C^\diamondsuit$ is s-holomorphic, then $F$ lies in the kernel of the discrete
$\overline{\partial}$-operator $\overline\partial_D\colon\C^B\to\C^W$ associated with the double graph $D_\G$ with bipartite structure $B=\diamondsuit$ and $W=V(\G)\cup V(\G^*)$.
The first statement now follows from~\cite[Theorem 2.5]{Cim1} applied to the bipartite graphs $D_{\G_n}$.
By Corollary~\ref{cor:s} above, $\widetilde{T'}F_n$ lies in the kernel of the discrete $\overline{\partial}$-operator $\overline\partial^{\varphi_n}_{C_n}\colon\C^{B_n}\to\C^{W_n}$
associated with the graph $C_{\G_n}$, where $\varphi_n$ is a discrete spin structure on $C_{\G_n}$ which can be chosen to represent the fixed spin structure $L$. The second statement
now follows from~\cite[Theorem 3.12]{Cim1} applied to the bipartite graphs $C_{\G_n}$. (Note that even though some rhombi of $C_{\G_n}$ are degenerate, Lemma 2.7 of~\cite{Cim1}
still holds in this setting, and the theorem does apply.)
\end{proof}

\subsection{The inverse Kac-Ward operator}
\label{sub:inv}

Let $(\G,x)\subset\SI$ be an arbitrary surface graph, and as before, let us consider a fixed vector field $\lambda$ on $\SI$ with zeroes of even index in $\SI\setminus\G$.
Recall that the associated Kac-Ward operator
$\KW=\KW^\lambda\colon\C^\EE\to\C^\EE$ can be defined by $(\KW^\lambda f)(e)=\sum_{e'\in\EE}\KW^\lambda(e,e')f(e')$ for $f\in\C^\EE$ and $e\in\EE$, with coefficients
\[
\KW^\lambda(e,e')=
\begin{cases}
1&\text{if $e=e'$;}\\ 
-x_e\exp\left(\textstyle{\frac{i}{2}}\alpha_\lambda(e,e')\right)&\text{if $o(e')=t(e)$ but $e'\neq \bar{e}$;}\\
0&\text{else,}
\end{cases}
\]
with $\alpha_\lambda(e,e')$ the rotation angle in radians of the velocity vector field along $e$ followed by $e'$ with respect to the vector field $\lambda$.

As mentioned in subsection~\ref{sub:KW}, its determinant is the square of a polynomial in the weight variables. The precise result is most conveniently stated using the terminology
of homology, that we now very briefly recall. Given a surface graph $\G\subset\SI$, let $C_0$ (resp. $C_1$, $C_2$) denote the $\Z_2$-vector space with basis the set of vertices (resp. edges, faces)
of $\G\subset\SI$. Elements of $C_i$ are called {\em $i$-chains\/}. Also, let $\partial_2\colon C_2\to C_1$ and $\partial_1\colon C_1\to C_0$ denote the {\em boundary operators\/} defined
in the obvious way. Since $\partial_1\circ\partial_2$ vanishes, the space of {\em $1$-cycles\/} $\mathrm{ker}(\partial_1)$ contains the space $\partial_2(C_2)$ of
{\em $1$-boundaries\/}. The {\em first homology space\/} $H_1(\SI;\Z_2):=\mathrm{ker}(\partial_1)/\partial_2(C_2)$ turns out not to depend on $\G$, but only on $\SI$: it has dimension $2g$, where $g$ is the genus of the closed
connected orientable surface $\SI$. Note that the intersection of curves defines a non-degenerate bilinear form on $H_1(\SI;\Z_2)$, that will be denoted by
$(\alpha,\beta)\mapsto\alpha\cdot\beta$.

We shall also need the following classical result of Johnson~\cite{Joh}: given a vector field $\lambda$ on $\SI$ with zeroes of even index and a piecewise smooth curve $\gamma$ in $\SI$ avoiding
the zeroes of $\lambda$, let $\rot_\lambda(\gamma)\in 2\pi\Z$ denote the rotation angle of the velocity vector field of $\gamma$ with respect to $\lambda$. Then, given
a homology class $\alpha\in H_1(\SI;\Z_2)$ represented by the disjoint union of oriented simple closed curves $\gamma_j$, the equality
\[
(-1)^{q_\lambda(\alpha)}=\prod_j-\exp\left(\textstyle{\frac{i}{2}}\rot_\lambda(\gamma_j)\right)
\]
gives a well-defined {\em quadratic form\/} on $H_1(\SI;\Z_2)$, i.e. a map $q_\lambda\colon H_1(\SI;\Z_2)\to\Z_2$ such that
$q_\lambda(\alpha+\beta)=q_\lambda(\alpha)+q_\lambda(\beta)+\alpha\cdot\beta$ for all $\alpha,\beta\in H_1(\SI;\Z_2)$. This implies in particular that,
for any oriented closed curve $\gamma$ with $t(\gamma)$ transverse self-intersection points,
\[
-\exp\left(\textstyle{\frac{i}{2}}\rot_\lambda(\gamma)\right)=(-1)^{q_\lambda(\gamma)+t(\gamma)}.
\]
(This can be checked by smoothing out the intersection points according to the orientation of $\gamma$ and using the definition of $q_\lambda$.)

Coming back to the Kac-Ward determinant, it was showed in~\cite{Cim2} that
\[
\det\big(\KW^\lambda\big)=\Big(\sum_{\xi\in\E}(-1)^{q_\lambda(\xi)}x(\xi)\Big)^2,
\]
where $\E$ denotes the set of 1-cycles in $\G$, that is, the set of subgraphs $\xi$ of $\G$ such that each vertex of $\G$ is adjacent to an even number
of edges of $\xi$, and $x(\xi)$ stands for $\prod_{e\in\xi}x_e$. Since $q_\lambda$ is well-defined in homology, this equality can be rewritten
\[
\det\big(\KW^\lambda\big)^{1/2}=\sum_{\alpha\in H_1(\SI;\Z_2)}(-1)^{q_\lambda(\alpha)}\sum_{\genfrac{}{}{0pt}{}{\xi\in\E}{[\xi]=\alpha}}x(\xi),
\]
which easily leads to the generalized Kac-Ward formula~\eqref{eqn:KW} for $Z^{\mathit{Ising}}=\sum_{\xi\in\E}x(\xi)$.

Let us introduce some more notation. Let $\widetilde\Gamma$ be the graph obtained from $\G$ by adding the vertex $z_e$ in the middle of each edge
$e\in E$. Given $e,e'\in\EE$, we shall denote by $\mathcal{E}(e,e')$ the set of subgraphs $\xi\subset\widetilde\G$ that contain the half-edges $(z_e,t(e))$ and
$(o(e'),z_{e'})$, and such that each vertex in $V(\widetilde\G)\setminus\{z_e,z_{e'}\}$ is adjacent to an even number of edges of $\xi$.
Given any $\xi\in\mathcal{E}(e,e')$, one can resolve its crossings to obtain a disjoint union $\gamma^0_\xi\sqcup\gamma_\xi$, where $\gamma^0_\xi$
consists in a family of disjoint simple closed curves and $\gamma_\xi$ is an oriented simple curve from $z_e$ to $z_{e'}$.
Note that the complex number $(-1)^{q_\lambda(\gamma_\xi^0)}\exp\left(\textstyle{\frac{i}{2}}\rot_\lambda(\gamma_\xi)\right)$ is unchanged by merging curves
in the family above; therefore, it does not depend on the choice of the smoothing of $\xi$.

\begin{definition}
\label{def:F}
Let $F^\lambda\colon\C^\EE\to\C^\EE$ be defined by $(F^\lambda g)(e)=\sum_{e'\in\EE}F^\lambda(e,e')g(e')$ for $g\in\C^\EE$ and $e\in\EE$, with coefficients
\[
F^\lambda(e,e')=
\begin{cases}
\sum_{e\notin\xi\in\E}(-1)^{q_\lambda(\xi)}x(\xi)&\text{if $e=e'$;}\\ 
\sum_{\xi\in\mathcal{E}(e,e')}(-1)^{q_\lambda(\gamma_\xi^0)}\exp\left(\textstyle{\frac{i}{2}}\rot_\lambda(\gamma_\xi)\right)x_ex(\xi)&\text{else,}
\end{cases}
\]
where $x(\xi)$ stands for the product of the weights of all the edges of $\G$ contained in $\xi$.
\end{definition}

We are now ready to state the main result of this paragraph. Note that in the planar case, it was obtained independently (and announced first) by Lis~\cite{Lis}.

\begin{theorem}
\label{thm:inv}
For any $\lambda$, $\KW^\lambda\circ F^\lambda=\det\big(\KW^\lambda\big)^{1/2}\cdot\mathit{Id}_{\C^\EE}$.
\end{theorem}

The following elementary lemma will be useful.

\begin{lemma}
\label{lemma:n}
Let $\xi$ be a graph containing an edge $e_1$ adjacent to a vertex $v$ of even degree. For any edge $e\neq e_1$ adjacent to $v$, let $n_\xi(e_1,e)$ denote
the number of edges of $\xi$ adjacent to $v$ strictly between $e_1$ and $e$ (on one given side). Then, $\sum_{v\in e\neq e_1}(-1)^{n_\xi(e_1,e)}=1$.
\end{lemma}
\begin{proof}
First note that the degree of $v$ being even, the parity of $n_\xi(e_1,e)$ does not depend on the choice of the side. Enumerating the edges $e_1,e_2,\dots,e_{2d}$ cyclically around $v$, we obtain $\sum_{v\in e\neq e_1}(-1)^{n_\xi(e_1,e)}=\sum_{j=2}^{2d}(-1)^{j-2}=1$.
\end{proof}

\begin{proof}[Proof of Theorem~\ref{thm:inv}]
Let $\lambda$ be a fixed vector field, that will be dropped from the notations for simplicity. We need to check that the coefficients of $\KW$ and of $F$
satisfy the relation
\[
\sum_{e\in\EE}\KW(e_1,e)F(e,e_2)=\delta_{e_1,e_2}\sum_{\xi\in\E}(-1)^{q(\xi)}x(\xi)
\]
for all $e_1,e_2\in\EE$. We will distinguish three cases.

Let us first assume that $e_1$ and $e_2$ coincide. In this case, $\sum_{e\in\EE}\KW(e_1,e)F(e,e_1)$ is by definition equal to
\[
\sum_{\genfrac{}{}{0pt}{}{\xi'\in\E}{e_1\notin\xi'}}(-1)^{q(\xi')}x(\xi')-\sum_{\genfrac{}{}{0pt}{}{o(e)=t(e_1)}{e\neq\bar{e}_1}} \sum_{\xi\in\mathcal{E}(e,e_1)}(-1)^{q(\gamma^0_\xi)}\exp\left(\textstyle{\frac{i}{2}(\alpha(e_1,e)+\rot(\gamma_\xi))}\right)x_{e_1}x_ex(\xi).
\]
Given $\xi\in\mathcal{E}(e,e_1)$, let $\xi'$ (resp. $\gamma_\xi'$) denote the 1-cycle obtained from $\xi$ (resp. $\gamma_\xi$) by adding the half-edges $(z_{e_1},t(e_1))$ and $(o(e),z_e)$.
The assignment $\xi\mapsto\xi'$ defines a map from $\E(e,e_1)$ to the set of 1-cycles containing $e_1$, and it satisfies $x_{e_1}x_ex(\xi)=x(\xi')$.
Furthermore,
\[
-\exp\left(\textstyle{\frac{i}{2}(\alpha(e_1,e)+\rot(\gamma_\xi))}\right)=-\exp\left(\textstyle{\frac{i}{2}\rot(\gamma_\xi')}\right)=(-1)^{q(\gamma'_\xi)+t(\gamma_\xi')}.
\]
The homology class of $\gamma'_\xi$ satisfies $[\gamma'_\xi]=[\gamma^0_\xi]+[\xi']$. Since $q$ is a quadratic form, we get the following equalities modulo 2:
\[
q(\gamma^0_\xi)+q(\gamma'_\xi)+t(\gamma_\xi')=q(\xi')+\gamma^0_\xi\cdot\xi'+t(\gamma'_\xi)=q(\xi')+\gamma^0_\xi\cdot\gamma'_\xi+t(\gamma'_\xi)=q(\xi')+n_{\xi'}(e_1,e),
\]
where $n_{\xi'}(e_1,e)$ denotes the number of edges of $\xi'$ adjacent to $t(e_1)=o(e)=:v$ strictly between $e_1$ and $e$. Using Lemma~\ref{lemma:n}, we get
\begin{align*}
\sum_{e\in\EE}\KW(e_1,e)F(e,e_1)&=\sum_{\genfrac{}{}{0pt}{}{\xi'\in\E}{e_1\notin\xi'}}(-1)^{q(\xi')}x(\xi')+\sum_{\genfrac{}{}{0pt}{}{\xi'\in\E}{e_1\in\xi'}}(-1)^{q(\xi')}x(\xi')
\Big(\sum_{\genfrac{}{}{0pt}{}{e\in\xi'}{v\in e\neq e_1}}(-1)^{n_{\xi'}(e_1,e)}\Big)\\
	&=\sum_{\xi'\in\E}(-1)^{q(\xi')}x(\xi')
\end{align*}
as expected.

Let us now consider the case of $e_1\neq e_2$ with $t(e_1)\neq o(e_2)$. By definition,
\begin{align*}
\sum_{e\in\EE}\KW(e_1,e)F(e,e_2)=&
\sum_{\xi'\in\E(e_1,e_2)}(-1)^{q(\gamma^0_{\xi'})}\exp\left(\textstyle{\frac{i}{2}\rot(\gamma_{\xi'})}\right)x(\xi')\;+\cr
	&-\sum_{\genfrac{}{}{0pt}{}{o(e)=t(e_1)}{e\neq\bar{e}_1}} \sum_{\xi\in\mathcal{E}(e,e_2)}(-1)^{q(\gamma^0_\xi)}
\exp\left(\textstyle{\frac{i}{2}(\alpha(e_1,e)+\rot(\gamma_\xi))}\right)x_{e_1}x_ex(\xi).
\end{align*}
Given $\xi\in\E(e,e_2)$, let $\xi'$ (resp. $\gamma_\xi'$) denote the subgraph (resp. the oriented curve) obtained from $\xi$ (resp. $\gamma_\xi$) 
by adding the half-edges $(z_{e_1},t(e_1))$ and $(o(e),z_e)$. This defines a map from $\E(e,e_2)$ to $\E(e_1,e_2)$ such that $x_{e_1}x_ex(\xi)=x(\xi')$.
Furthermore,
\[
-\exp\left(\textstyle{\frac{i}{2}(\alpha(e_1,e)+\rot(\gamma_\xi)-\rot(\gamma_{\xi'})}\right)=-\exp\left(\textstyle{\frac{i}{2}\rot(\gamma)}\right)=(-1)^{q(\gamma)+t(\gamma)},
\]
where $\gamma$ is the closed oriented curve obtained by following $\gamma'_\xi$ and then $-\gamma_{\xi'}$. Since its homology class satisfies
$[\gamma]=[\gamma'_\xi+\gamma_{\xi'}]=[\gamma^0_\xi]+[\gamma^0_{\xi'}]$, we have
\[
q(\gamma^0_\xi)+q(\gamma)+t(\gamma)=q(\gamma^0_{\xi'})+\gamma^0_\xi\cdot\gamma^0_{\xi'}+t(\gamma)=q(\gamma^0_{\xi'})+n_{\xi'}(e_1,e).
\]
This leads to
\[
\sum_{e\in\EE}\KW(e_1,e)F(e,e_2)=
\sum_{\xi'\in\E(e_1,e_2)}(-1)^{q(\gamma^0_{\xi'})}\exp\left(\textstyle{\frac{i}{2}\rot(\gamma_{\xi'})}\right)x(\xi')
\Big(1-\sum_{\genfrac{}{}{0pt}{}{e\in\xi'}{t(e_1)\in e\neq e_1}}(-1)^{n_{\xi'}(e_1,e)}\Big),
\]
which vanishes by Lemma~\ref{lemma:n}.

Let us finally assume that $e_1\neq e_2$ and $t(e_1)=o(e_2)=v$. This time, we have
\[
\sum_{e\in\EE}\KW(e_1,e)F(e,e_2)=\exp\left(\textstyle{\frac{i}{2}\alpha(e_1,e_2)}\right)x_{e_1}\Big(X_1+X_2-\sum_{\genfrac{}{}{0pt}{}{\xi'\in\E}{e_2\notin\xi'}}(-1)^{q(\xi')}x(\xi')\Big),
\]
where
\[
X_1=\sum_{\xi\in\mathcal{E}(e_1,e_2)}(-1)^{q(\gamma^0_\xi)}\exp\left(\textstyle{\frac{i}{2}(-\alpha(e_1,e_2)+\rot(\gamma_\xi))}\right)x(\xi)
\]
and
\[
X_2=-\sum_{\genfrac{}{}{0pt}{}{o(e)=v}{e\neq\bar{e}_1,e_2}}\sum_{\xi\in\mathcal{E}(e,e_1)}(-1)^{q(\gamma^0_\xi)}\exp\left(\textstyle{\frac{i}{2}(-\alpha(e_1,e_2)+\alpha(e_1,e)+\rot(\gamma_\xi))}\right)x_ex(\xi).
\]
Let us deal with the first term. Given $\xi\in\mathcal{E}(e_1,e_2)$, let $\xi'\in\E$ (resp. $\gamma_\xi'$) denote the 1-cycle obtained from $\xi$ (resp. $\gamma_\xi$) by removing the half-edges $(z_{e_1,}v)$ and $(v,z_{e_2})$. We have $x(\xi)=x(\xi')$, while
\[
\exp\left(\textstyle{\frac{i}{2}(-\alpha(e_1,e_2)+\rot(\gamma_\xi))}\right)=-\exp\left(\textstyle{\frac{i}{2}\rot(\gamma_\xi')}\right)=(-1)^{q(\gamma'_\xi)+t(\gamma_\xi')}=(-1)^{q(\xi')+n_{\xi'}(e_1,e_2)}
\]
as in the first step of the proof. Therefore,
\[
X_1=\sum_{\genfrac{}{}{0pt}{}{\xi'\in\E}{e_1,e_2\notin\xi'}}(-1)^{q(\xi')+n_{\xi'}(e_1,e_2)}x(\xi').
\]
Let us now deal with the second term. This time, let us assign to $\xi\in\E(e,e_2)$ the 1-cycle $\xi'\in\E$ obtained from $\xi$ by adding the half-edge
$(v,z_e)$ and removing $(v,z_{e_2})$. Clearly, $x_ex(\xi)=x(\xi')$. As before, we obtain
\[
-(-1)^{q(\gamma^0_\xi)}\exp\left(\textstyle{\frac{i}{2}(-\alpha(e_1,e_2)+\alpha(e_1,e)+\rot(\gamma_\xi))}\right)=(-1)^{q(\xi')+n_{\xi'}(e_1,e)},
\]
so that
\[
X_2=\sum_{\genfrac{}{}{0pt}{}{\xi'\in\E}{e_2\notin\xi'}}(-1)^{q(\xi')}x(\xi')\Big(\sum_{\genfrac{}{}{0pt}{}{v\in e\in\xi'}{e\neq\bar{e}_1,e_2}}(-1)^{n_{\xi'}(e_1,e)}\Big).
\]
Using Lemma~\ref{lemma:n} one last time, we get
\[
X_1+X_2=\sum_{\genfrac{}{}{0pt}{}{\xi'\in\E}{e_2\notin\xi'}}(-1)^{q(\xi')}x(\xi')
\Big(\sum_{\genfrac{}{}{0pt}{}{v\in e\in\xi'}{e\neq\bar{e}_1}}(-1)^{n_{\xi'}(e_1,e)}\Big)=\sum_{\genfrac{}{}{0pt}{}{\xi'\in\E}{e_2\notin\xi'}}(-1)^{q(\xi')}x(\xi').
\]
Therefore, this coefficient vanishes as well, and the proof is completed.
\end{proof}

This theorem immediately provides us with a wealth of s-holomorphic functions.
These generalize the spin-Ising fermionic observables of~\cite{CS09} and of~\cite{C-I}, which correspond to the case of critical isoradial graphs embedded in
simply-connected planar domains, and critical square lattices embedded in multiply-connected planar domains, respectively.

\begin{corollary}
\label{cor:s-holo}
Given any fixed oriented edge $e_0\in\EE$, the function $F_{e_0}\in\C^{\diamondsuit}$ defined by
\[
F_{e_0}(z)=\frac{\exp\left(i\textstyle{\frac{\pi}{4}}\right)}{\cos(\theta_e/2)}\sum_{\xi\in\E(e_0,z)}(-1)^{q(\gamma_\xi^0)}\exp\left(-\textstyle{\frac{i}{2}\rot(\gamma_\xi)}\right)x(\xi),
\]
where $z=z_e\neq z_{e_0}$ and $\E(e_0,z)=\E(e_0,e)\cup\E(e_0,\bar{e})$, is s-holomorphic around every vertex of $\G$ not adjacent to $e_0$.
\end{corollary}
\begin{proof}
First note that there is a bijection between $\E(e_0,z)$ and $\E(z,\bar{e}_0)=\E(e,\bar{e}_0)\cup\E(\bar{e},\bar{e}_0)$ reversing the orientation of $\gamma_\xi$.
Therefore, for $z=z_e\neq z_{e_0}$,
\[
F_{e_0}(z)
=\frac{\exp\left(i\textstyle{\frac{\pi}{4}}\right)}{\sin(\theta_e/2)}\sum_{\xi\in\E(z,\bar{e}_0)}(-1)^{q(\gamma_\xi^0)}\exp\left(\textstyle{\frac{i}{2}\rot(\gamma_\xi)}\right)x_ex(\xi)
=\exp\left(i\textstyle{\frac{\pi}{4}}\right)(S^{-1}f)(z),
\]
with $f\in\C^\EE$ defined by $f(e)=F(e,\bar{e}_0)$. By Proposition~\ref{prop:ker},
$F_{e_0}$ is s-holomorphic around a vertex $v$ if and only if $(\KW f)(\bar{e})$ vanishes for all $\bar{e}\in\EE_v$.
For $v$ not adjacent to $e_0$, we have $(\KW f)(\bar{e})=(\KW\circ F)(\bar{e},\bar{e}_0)=0$ by Theorem~\ref{thm:inv}.
\end{proof}

This theorem also has a direct link to the Ising model. Indeed, consider a locally finite planar graph $\mathcal{G}$ with edge weights
$J=(J_e)_e\in(0,\infty)^{E(\mathcal{G})}$ and faces that are bounded topological discs. Let us assume that this weighted graph is {\em biperiodic\/},
i.e. invariant under the action of a lattice $L\simeq\Z^2$. 

\begin{corollary}
\label{cor:crit}
There exists a non-trivial biperiodic s-holomorphic function on the weighted graph $(\mathcal{G},x)$, where $x_e=\tanh(\beta J_e)$, if and only if $\beta$ is the
critical inverse temperature for the Ising model on $(\mathcal{G},J)$.
\end{corollary}
\begin{proof}
A biperiodic weighted graph $(\mathcal{G},x)$ as above is equivalent to a finite weighted surface graph $(\G,x)=(\mathcal{G}/L,x)$ embedded in the torus
$\C/L$. In the recent paper~\cite{C-D}, Duminil-Copin and the author identified the critical inverse temperature $\beta_c$ for the Ising model
on $(\mathcal{G},J)$ as the only value of $\beta$ such that $\det(\KW^\lambda(\G,x))$ vanishes. The result therefore follows from Theorems~\ref{thm:s} and~\ref{thm:inv}.
\end{proof}

\subsection{The integral of $F^2$}
\label{sub:F2}

The aim of this last paragraph is to show that one of the most remarkable facts discovered in~\cite{CS09} extends to our setting: the existence of a discrete version of
$\Im\int f(z)^2\mathit{dz}$ for s-holomorphic functions.

Let us fix a weighted surface graph $(\G,x)\subset\SI$ with weights parametrized by $x_e=\tan(\theta_e/2)$, and let us assume that dual edges $e$ and $e^*$ meet orthogonally, so that $D_{e^*}=iD_e$.
Given an s-holomorphic function $F\in\C^\diamondsuit$, define $H\in\C^\Lambda=\C^{V(\G)}\oplus\C^{V(\G^*)}$ by
\[
H(v)-H(w):=2\left|\Pr\left(F(z);\left[i\exp(i(a_e+\theta_e))\right]^{-\frac{1}{2}}\right)\right|^2=2\left|\Pr\left(F(z');\left[i\exp(i(a_{e'}-\theta_{e'}))\right]^{-\frac{1}{2}}\right)\right|^2
\]
for $v\in V(\G)$ and $w\in V(\G^*)$, with the notations illustrated below. Note that the equality above follows from the very definition of s-holomorphicity.

\begin{figure}[h]
\labellist\small\hair 2.5pt
\pinlabel {$v$} at -30 200
\pinlabel {$w$} at 245 190
\pinlabel {$e$} at 300 50
\pinlabel {$e'$} at 300 280
\pinlabel {$z$} at 136 74
\pinlabel {$z'$} at 150 300
\endlabellist
\begin{center}
\includegraphics[height=2.5cm]{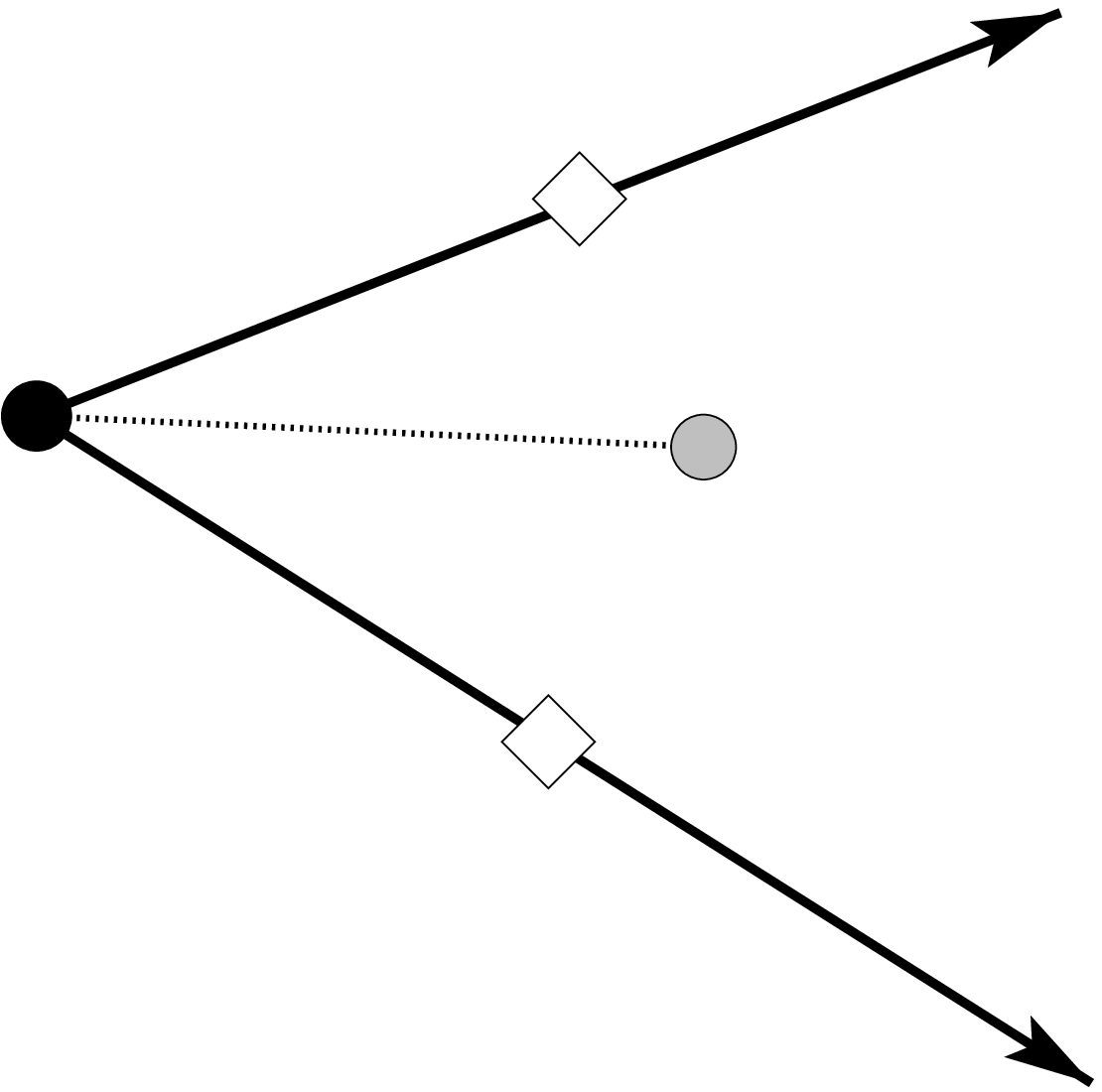}
\end{center}
\end{figure}

\begin{proposition}
\label{prop:F2}
The restriction of $H\in\C^\Lambda$ to any simply connected domain in $\SI$ is well-defined up to an additive constant. Moreover, it satisfies the equation
\[
H(v_2)-H(v_1)=\Im\left(2\cos(\theta_e)D_e\,F(z_e)^2\right)
\]
for any $v_1,v_2\in V(\G)$ (resp. $V(\G^*)$) linked by $e=(v_1,v_2)\in\EE(\G)$ (resp. $\EE(\G^*)$).
\end{proposition}
\begin{proof}
Fix $z=z_e\in\diamondsuit$, and let $\theta=\theta_e$ denote the corresponding angle. Also, let $(v_1,w_1,v_2,w_2)$ denote the four vertices in $\Lambda$ around $z$ in counterclockwise order
with $e=(v_1,v_2)\in\EE(\G)$, so that $e^*=(w_1,w_2)\in\EE(\G^*)$. By definition of $H$ and of $D_e$, we have
\begin{align*}
H(v_2)-&H(v_1)=(H(v_2)-H(w_2))-(H(v_1)-H(w_2))\cr
	&=2\left|\Pr\left(F(z)\,;\,\left[i\exp(i(a_{\bar{e}}-\theta))\right]^{-1/2}\right)\right|^2-2\left|\Pr\left(F(z)\,;\,\left[i\exp(i(a_e+\theta))\right]^{-1/2}\right)\right|^2\cr
	&=2\,\Re\left((i\exp(-i\theta))^{1/2}D_{\bar{e}}^{1/2}F(z)\right)^2-2\,\Re\left((i\exp(i\theta))^{1/2}D_e^{1/2}F(z)\right)^2\cr
	&=\frac{F(z)^2}{2}\left(i\exp(-i\theta)D_{\bar{e}}-i\exp(i\theta)D_e\right)+\frac{\overline{F(z)}^2}{2}\left(-i\exp(i\theta)\overline{D_{\bar{e}}}+i\exp(-i\theta)\overline{D_e}\right)\cr
	&=\cos(\theta)\left(-i\,D_e\,F(z)^2+i\,\overline{D_e}\,\overline{F(z)}^2\right)\cr
	&=\Im\left(2\cos(\theta)D_e\,F(z)^2\right),
\end{align*}
using the fact that $D_{\bar{e}}=-D_e$. A similar computation leads to
\[
(H(v_2)-H(w_1))-(H(v_1)-H(w_1))=\Im\left(2\cos(\theta)D_e\,F(z)^2\right).
\]
By simple connectivity, this shows that $H$ is well-defined on $V(\G)$ up to an additive constant. Similarly, one computes
\[
(H(v_1)-H(w_1))-(H(v_1)-H(w_2))=\Re\left(2\sin(\theta)D_e\,F(z)^2\right)
\]
and
\[
(H(v_2)-H(w_1))-(H(v_2)-H(w_2))=\Re\left(2\sin(\theta)D_e\,F(z)^2\right).
\]
Since $e^*=(w_1,w_2)$, $\theta_{e^*}=\frac{\pi}{2}-\theta$ and $D_{e^*}=iD_e$, we have
\[
H(w_2)-H(w_1)=\Im\left(2\cos(\theta_{e^*})D_{e^*}\,F(z)^2\right).
\]
This completes the proof.
\end{proof}

Let us conclude this section, and this paper, with one last remark and some comments about possible future work.
In the isoradial case, it is showed in~\cite{CS09} that $H$ is discrete superharmonic on $\G$ and subharmonic on $\G^*$.
More precisely, given any vertex $v\in V(\G)$, let us denote by $\theta_1,\dots,\theta_n$ the adjacent half-rhombus angles. Then,
\[
(\Delta_\G H)(v)=-\frac{1}{\mu_v}Q^{(n)}_{\theta_1;\dots;\theta_n}(t_1,\dots,t_n),
\]
where $Q^{(n)}_{\theta_1;\dots;\theta_n}$ is some explicit quadratic form depending on the parameters $\theta_1,\dots,\theta_n$, and $t_1,\dots,t_n$ are real numbers.
For $\theta_1,\dots,\theta_n>0$ with $\theta_1+\dots+\theta_n=\pi$, Chelkak-Smirnov prove that this quadratic form is non-negative, which implies the claim.

It turns out that even in the most general case of an arbitrary weighted surface graph, the equality displayed above still holds, where $\Delta_\G$ is the associated discrete Laplacian
(Definition~\ref{def:Delta}). Moreover, $Q^{(n)}_{\theta_1;\dots;\theta_n}$ is easily seen to remain non-negative when the $\theta_j$'s add up to an odd multiple of $\pi$.
Therefore, $H$ remains discrete superharmonic on $\G$ and subharmonic on $\G^*$ when the weighted graph $(\G,x)$ can be isoradially embedded in a flat surface
(recall Definition~\ref{def:iso}) with the given weights $x$ coinciding with the corresponding critical ones. This could possibly lead to some conformal invariance
results in the critical isoradial case in non-trivial topology.
However, the quadratic form $Q^{(n)}_{\theta_1;\dots;\theta_n}$ is obviously not non-negative for arbitrary (positive) values of the parameters $\theta_1,\dots,\theta_n$.
Therefore, the discrete superharmonicity of $H$ does {\em not\/} extend in full generality. In particular, the extension of the results of Chelkak-Smirnov to arbitrary biperiodic
graphs at criticality will require the use of another (less naive) discretization of the Laplacian.

%%%%%%%%%%%%%%%%%%%%%%

\bibliographystyle{plain}

\bibliography{Ising}

\end{document}